\documentclass[sigconf]{acmart}
\usepackage{amsmath}

\usepackage{amssymb}
\usepackage{bm}
\usepackage{nicefrac}
\usepackage{booktabs}
\usepackage{array}
\usepackage{multirow}
\usepackage{threeparttable}
\usepackage{makecell}
\usepackage[procnumbered,ruled,vlined,linesnumbered]{algorithm2e}
\usepackage{siunitx}
\usepackage{stfloats}
\usepackage{graphicx}
\usepackage{subfigure}
\usepackage{hyperref}

\newtheorem{problem}{Problem}
\newtheorem{theorem}{Theorem}[section]

\newtheorem{lemma}[theorem]{Lemma}

\newtheorem{proposition}[theorem]{Proposition}

\newtheorem{definition}[theorem]{Definition}

\newenvironment{fminipage}%
{\begin{Sbox}\begin{minipage}}%
		{\end{minipage}\end{Sbox}\fbox{\TheSbox}}

\def\defeq{\stackrel{\mathrm{def}}{=}}
\def\setof#1{\left\{#1  \right\}}
\def\sizeof#1{\left|#1  \right|}

\def\eps{\epsilon}

\def\abs#1{\left|#1  \right|}

\def\trace#1{\mathrm{Tr} \left(#1 \right)}
\def\norm#1{\left\| #1 \right\|}

\newcommand{\D}{\ensuremath{\mathbf{D}}}

\newfont{\nset}{msbm10}

\def\norm#1{\Vert #1 \Vert}

\def\kh#1{\left( #1 \right)}

\def\ceil#1{\left\lceil #1 \right\rceil}

\def\defeq{\stackrel{\mathrm{def}}{=}}

\newcommand{\removelatexerror}{\let\@latex@error\@gobble}

\newcommand{\rea}{\mathbb{R}}

\newcommand\LL{\bm{\mathit{L}}}
\newcommand\Otil{\widetilde{O}}

\def\defeq{\stackrel{\mathrm{def}}{=}}
\def\trace#1{\mathrm{Tr} \left(#1 \right)}
\def\sizeof#1{\left|#1  \right|}
\def\setof#1{\left\{#1  \right\}}

\newcommand{\GainsEst}{\textsc{GainsEst}}
\newcommand{\FGainsEst}{\textsc{F-GainsEst}}
\newcommand{\wmax}{w_{{\max}}}

\newcommand{\ExactSM}{\textsc{Exact}}
\newcommand{\ApproxiSM}{\textsc{Approx}}

\newcommand\WW{\boldsymbol{\mathit{W}}}
\newcommand\XX{\boldsymbol{\mathit{X}}}

\newcommand\xx{\boldsymbol{\mathit{x}}}
\newcommand\hh{\boldsymbol{\mathit{h}}}
\newcommand\aaa{\boldsymbol{\mathit{a}}}
\newcommand\zeov{\boldsymbol{\mathit{0}}}
\newcommand\bb{\boldsymbol{\mathit{b}}}
\newcommand\cc{\boldsymbol{\mathit{c}}}
\newcommand\ee{\boldsymbol{\mathit{e}}}
\newcommand\pp{\boldsymbol{\mathit{p}}}
\newcommand\qq{\boldsymbol{\mathit{q}}}
\newcommand\rr{\boldsymbol{\mathit{r}}}
\renewcommand\SS{\boldsymbol{\mathit{S}}}
\renewcommand\AA{\boldsymbol{\mathit{A}}}
\newcommand\BB{\boldsymbol{\mathit{B}}}
\newcommand\CC{\boldsymbol{\mathit{C}}}
\newcommand\JJ{\boldsymbol{\mathit{J}}}
\newcommand\KK{\boldsymbol{\mathit{K}}}
\newcommand\DD{\boldsymbol{\mathit{D}}}
\newcommand\HH{\boldsymbol{\mathit{H}}}
\newcommand\EE{\boldsymbol{\mathit{E}}}
\newcommand\PP{\boldsymbol{\mathit{P}}}
\newcommand\MM{\boldsymbol{\mathit{M}}}
\newcommand\ZZ{\boldsymbol{\mathit{Z}}}
\newcommand\RR{\boldsymbol{\mathit{R}}}
\newcommand\QQ{\boldsymbol{\mathit{Q}}}

\newcommand\II{\boldsymbol{\mathit{I}}}
\renewcommand\vv{\boldsymbol{\mathit{v}}}
\newcommand{\SDDMSolver}{\textsc{Solve}}
\newcommand\ZZtil{\widetilde{\boldsymbol{\mathit{Z}}}}
\newcommand\zztil{\widetilde{\boldsymbol{\mathit{z}}}}

\DontPrintSemicolon
\SetKw{KwAnd}{and}
\SetFuncSty{textsc}
\SetKwInOut{Input}{Input\ \ \ \ }
\SetKwInOut{Output}{Output}

\usepackage{tabularx}
\usepackage{stfloats}

\begin{document}
	
	\title{Minimizing Polarization in Noisy Leader-Follower \\Opinion Dynamics}
	
	\author{Wanyue Xu}
	\affiliation{%
		\institution{Fudan University}
		\city{Shanghai}
		\country{China}}
	\email{xuwy@fudan.edu.cn}
	
	\author{Zhongzhi Zhang}
    \authornote{Corresponding author.}
	\affiliation{%
		\institution{Fudan University}
		\city{Shanghai}
		\country{China}}
	\email{zhangzz@fudan.edu.cn}
	
	\renewcommand{\shorttitle}{Minimizing Polarization in Noisy Leader-Follower Opinion Dynamics}
	
	\begin{abstract}
The operation of creating edges has been widely applied to optimize relevant quantities of opinion dynamics. In this paper, we consider a problem of polarization optimization for the leader-follower opinion dynamics in a noisy social network with $n$ nodes and $m$ edges, where a group $Q$ of $q$ nodes are leaders, and the remaining $n-q$ nodes are followers. We adopt the popular leader-follower DeGroot model, where the opinion of every leader is identical and remains unchanged, while the opinion of every follower is subject to white noise. The polarization is defined as the steady-state variance of the deviation of each node's opinion from leaders' opinion, which  equals one half of the effective resistance $\mathcal{R}_Q$    between  the node group $Q$ and all other nodes. Concretely, we propose and study the problem of minimizing  $\mathcal{R}_Q$ by adding $k$ new edges with each incident to a node in $Q$. We show that the objective function is monotone and supermodular. We then propose a simple greedy algorithm with an approximation factor $1-1/e$ that approximately solves the problem in $O((n-q)^3)$ time. To speed up the computation, we also provide a fast algorithm to compute $(1-1/e-\eps)$-approximate effective resistance $\mathcal{R}_Q$, the running time of which is $\Otil (mk\eps^{-2})$ for any $\eps>0$, where the $\Otil (\cdot)$ notation suppresses the ${\rm poly} (\log n)$ factors. Extensive experiment results show that our second algorithm  is both effective and efficient.

\end{abstract}

	
	
	
\maketitle

\section{Introduction}

The rapid development of digital technology and the Internet leads to an explosive growth of online social networks and social media~\cite{Le20}, which have drastically changed people's work, health, and life~\cite{SmCh08}. For example, the enormous popularity of online social networks and social media have brought great convenience to people all over the world to exchange in real time their opinions on some important hot issues or topics, which results in a fundamental change of the ways of opinion propagation, share and formation~\cite{PeRo19,AnYe19,BuLiZhCaLiSh20}. At the same time, the wide-range usage of online social networks and social media also exacerbates some social phenomena in the online virtual world, such as polarization~\cite{MaTeTs17,MuMuTs18,XuBaZh21} and disagreement~\cite{GaKlTa20}, although these phenomena might exist in human societies millennia ago.


In order to understand the mechanisms for opinion transmission, evolution, and shaping, as well as their resulting social phenomena aggravated in virtual space, a variety of models have been developed, among which the DeGroot model~\cite{De74} is probably the first discrete-time model for opinion dynamics, the continuous-time counterpart of which was introduced in~\cite{Ta68}. After their establishment, the original discrete-time and continuous-time DeGroot models have been modified or extended  by incorporating different factors affecting opinion dynamics~\cite{No20}, such as stubborn individuals~\cite{FrJo90} and noise~\cite{XiBoKi07}. When an individual is completely stubborn~\cite{XuZhGuZhZh20}, it is a leader who never changes its opinion. In~\cite{PaBa10}, a noisy leader-follower model for opinion dynamics was proposed and studied, where some nodes are leaders with identical opinion, while the remaining nodes are followers, which are influenced by leaders and are simultaneously subject to stochastic disturbances.

In the above noisy leader-follower model, the presence of noise prevents followers from reaching consensus, whose opinions fluctuate around the opinion of leaders in long time limit. The derivation of followers' opinions from that of leaders can be quantified by coherence~\cite{PaBa10}, which is similar to the measure of polarization introduced  in~\cite{MuMuTs18}. In the context of opinion dynamics, polarization describes the extent of macroscopic deviations of opinions in the social system~\cite{MuMuTs18}. In different networks, polarization can exhibit rich behavior dependent on the position of leaders and network topology. Many previous work considered the problem of how to choose leaders in order to optimize relevant quantities for the noisy leader-follower model, such as minimizing polarization~\cite{PaBa10, ClPo11} that measures the performance or role of leaders in the social system. Except the operation on nodes,  the operation at the edge level has not been much studied, at least in the context of polarization optimization for noisy leader-follower opinion dynamics, despite  that the edge operation has been widely used in various application scenarios~\cite{IsErTeBe12,PaPiTs16,LiPaYiZh20,ZhZhCh21},  
even in the context of opinion dynamics~\cite{BiKlOr11,GaDeGiMa17,ChLiDe18,AmSi19}.

Inspired by existing work on edge operation, in this paper, we focus on the problem of optimizing polarization by creating edges. Specifically, we study an optimization problem about noisy leader-follower DeGroot model~\cite{PaBa10} of opinion dynamics on a graph $G=(V,E)$ with $n$ nodes and $m$ edges. In the model, a group $Q \subset V$ of $q$ nodes are leaders with a fixed opinion, while all other nodes are followers, which are exposed to noise. The problem we address  is how to optimally create $k$ (a positive integer) new edges with each being incident to a node in $Q$, so that the polarization $P_{Q}(G) $ of the opinion dynamics is maximized. Based on the relation that $P_{Q}(G) $ is equal to one half of the resistance distance $\mathcal{R}_Q$ between node group $Q$ and all nodes, the problem is reduced to minimizing $\mathcal{R}_Q$ by adding $k$ edges connecting nodes in $Q$.

In addition to formulating the problem, other main contributions of our work include the following three points. First, it is shown that the objective function of optimization problem is monotone and supermodular. Then, two greedy approximation algorithms are developed to minimize the quantity $\mathcal{R}_Q$, by iteratively building $k$ edges. The former is a $ (1-1/e)$-approximation algorithm, while the latter is a $ (1-1/e-\eps)$-approximation algorithm for any small $\eps>0$. The time complexity of the two algorithms are, respectively, $O((n-q)^3)$ and $\Otil (mk\eps^{-2})$, where the $\Otil (\cdot)$ notation hides ${\rm poly} (\log n)$ factors. Finally, the performance of our algorithms are tested on various real networks, which substantially reduce the resistance distance $\mathcal{R}_Q$ and outperform several other baseline strategies of adding edges.

\section{Related Work}

In this section, we review  some existing work lying close to ours.


\textit{Models.} DeGroot model~\cite{De74} is one of the most popular models for opinion dynamics. Though simple and succinct, it captures some important aspects of the interactions and processes of opinion evolution. The DeGroot model and its continuous-time counterpart~\cite{Ta68} are the basis of various subsequent models for opinion dynamics. For example, the Friedkin-Johnson (FJ) model~\cite{FrJo90} is an important extension of the DeGroot model by incorporating the intrinsic opinion and susceptibility to persuasion~\cite{AbKlPaTs18} for every node. Another major modification of the DeGroot model is the Altafini model~\cite{Al13}, by considering both cooperative and antagonistic interactions between nodes. Moreover, the DeGroot model was also extended by taking into account the influence of noise on the opinion evolution of every individual~\cite{XiBoKi07,BaJoMiPa12}. For other modifications or extensions of the DeGroot model, we refer the readers to the review literature~\cite{DoDiMa17}. 


\textit{Optimization Problems.} Apart from the variants of the DeGroot model, another active direction about the extension of the DeGroot model is the optimization problem for different objectives. Many authors introduced leaders into DeGroot model, where every leader is totally stubborn with identical or different opinions. For the case that there are two types of leaders with opposite opinions, various leader placement problems were formulated and studied, in order to optimize different objectives, such as maximizing the opinion diversity~\cite{MaPa19,YiCaPa21} and maximizing the overall opinion~\cite{VaFaFr14}. For the case that leaders' opinions are the same, some similar combinatorial optimization problems were proposed for different purposes, such as minimizing convergence error~\cite{ClAlBuPo14} or convergence rate~\cite{LiXuLuChZe21}. 
Particularly, in~\cite{XiBoKi07} homogenous leaders were incorporated to the DeGroot model to affect the followers, which are subject to noise. In the noisy leader-follower model, the limiting opinions of followers fluctuate around leaders' opinion, which can be quantified by coherence or polarization. Since polarization depends on  leader position and network structure, the authors   studied the problem of optimally placing leaders to minimize polarization. Different from previous perspectives, we consider the problem of designing an ingenious strategy to add edges to minimize polarization.  


\textit{Graph Edit by Adding Edge.} Note that the quantity polarization can be considered as a measure of the role or centrality of leaders in the noisy leader-follower opinion dynamics. From this point of view, the essence for the addressed problem of minimizing polarization by adding edges is to increase the influence of leaders by creating edges. Thus far, concerted efforts have been devoted to the problem of optimizing the centrality of a node group by adding $k$ edges connecting nodes in the group. For example, many scientists have tackled the problem of adding a fixed number of edges to maximize different centrality measures of node groups, including group betweennees~\cite{MeSiSiBaSw18}, group closeness~\cite{OhSaKiMo17}, and group coverage~\cite{MeSiSiBaSw18}, and so on. However, previous work do not consider the optimization problem for polarization by edge operation.


In the context of opinion dynamics, various optimization problems were also formulated to achieve different goals by edge operations, especially by adding edges. In~\cite{BiKlOr11}, edge addition strategy was adopted to minimize the social cost at equilibrium in the FJ model. In~\cite{GaDeGiMa17} and~\cite{ChLiDe18}, the operation of creating links was used to reduce controversy and risk of conflict, respectively. In~\cite{AmSi19}, a limited number of links was strategically recommended against malicious control of user opinions. In~\cite{ZhZh21,ZhZhLiZh23}, edge addition method was introduced to maximize the overall opinion. Finally, in~\cite{ZhBaoZh21}, ingenious strategy of link recommendation was to designed to minimize the sum of polarization and disagreement in  social networks. To the best of our knowledge, we are the first to adopt the strategy of edge addition to minimize polarization in noisy leader-follower opinion dynamics.

\section{Preliminaries}

In this section, we give a brief  introduction to some useful notations and tools, in order to facilitate  the description of our problem and algorithms.

\subsection{Notations}

We use normal lowercase letters like $a,b,c$ to represent scalars in $\rea$, normal uppercase letters like $A,B,C$ to represent sets, bold lowercase letters like $\aaa,\bb,\cc$ to represent vectors, and bold uppercase letters like $\AA,\BB,\CC$ to represent matrices.
Let $\aaa^\top$ and $\AA^\top$  denote, respectively, transpose of  vector $\aaa$ and matrix  $\AA$. Let $\trace \AA$ denote the trace of matrix $\AA$. We write $\AA_{[i,j]}$ to denote the entry at $i^{\rm th}$ row and $j^{\rm th}$ column of $\AA$. We use  $\AA_{[i,:]}$ and $\AA_{[:,j]}$ to denote, respectively, the $i^{\rm th}$ row  and the $j^{\rm th}$ column of $\AA$.    We write sets in matrix subscripts to denote submatrices. For example, $\LL_{Q}$ denotes the submatrix of $\LL$ obtained from $\LL$ by removing both the  row and  column indices in $Q$.  For two matrices $\AA$ and $\BB$, we write $\AA \preceq \BB$ to denote
that $\BB - \AA$ is positive semidefinite, that is,  for every real vector $\xx$ the relation $\xx^\top \AA \xx \leq \xx^\top \BB \xx$ holds.

We continue to introduce the notion of $\eps$-approximation.
\begin{definition}
	Let $a$ and $b$ be two nonnegative scalars.
	We say $a$ is an $\eps$-approximation ($0\leq\eps\leq1/4$) of $b$ if
	\begin{align*}
		(1-\eps) a \leq b \leq (1+\eps) a.
	\end{align*}
\end{definition}
For simplicity, we use the notation $a \approx_{\eps} b$ to denote that $a$ is an $\eps$-approximation of $b$.
There are some basic properties for $\eps$-approximation.
\begin{proposition}
	For nonnegative scalars $a,b,c$, and $d$,
	\begin{enumerate}
		\item \label{apxstart} if $a \approx_\eps b$, then $a + c \approx_\eps b + c$;
		\item if $a \approx_\eps b$ and $c \approx_\eps d$, then $a + c \approx_\eps b + d$;
		\item if $a \approx_{\eps} b$ and $c \approx_{\eps} d$, then $a/c \approx_{3\eps} b/d$;
	\end{enumerate}
\end{proposition}

Since we use the supermodularity in our algorithm, we first give a definition of the supermodular function. Let $X$ be a finite set, and $2^X$ be the set of all subsets of $X$. Then a supermodular function can be defined as follows.
\begin{definition}
	Let $f: 2^X \to \mathbb{R}$ be a set function on $X$. For any subset $S \subset T \subset X$ and any element $a \in X \setminus T$ , the function $f$ is supermodular if it satisfies
	\begin{align}
		f(S) - f(S \cup \{a\}) \geq f(T) - f(T \cup \{a\}).
	\end{align}
\end{definition}
We also give the definition of a monotone set function.
\begin{definition}
	A set function $f: 2^X \to \mathbb{R}$ is monotone decreasing if for any subset $S \subset T \subset X$ ,
	\begin{align}
		f(S) > f(T).
	\end{align}
\end{definition}

\subsection{Graphs and Laplacian Matrix}

Let $G = (V,E,w)$ be a connected undirected weighted network (graph), where $V$ is the set of vertices/nodes, $E \subseteq V \times V$ is the set of edges, and $w : E \to \rea_{+}$ is the edge weight function. Let $w_{\rm max}$ and $w_{\rm min}$ denote, respectively, the maximum and minimum weight among all edges.  For a pair of vertices $u,v\in E$, we write $u\sim v$ to denote $(u,v) \in E$. For a vertex  $u$, its weighted degree $\mathrm{deg}(u)$ is  $\sum_{u\sim v} w(u,v)$. Let $n = |V|$ denote the number of vertices and $m = |E|$ denote the number of edges.  

The Laplacian matrix $\LL$ of $G$ is an $n\times n$ matrix,  whose entry associated with $u^{\mathrm{th}}$ row and $v^{\mathrm{th}}$ is defined as: $\LL_{[u,v]} =- w(u,v)$ if  $u\sim v$, $\LL_{[u,v]} =\mathrm{deg}(u)$ if  $u=v$, and $\LL_{[u,v]} =0$ otherwise. If we fix an arbitrary orientation for all  edges in $G$, then we can define the signed edge-vertex incidence matrix  $\BB_{m\times n}$ of  graph $G$,  whose entries are:  $\BB_{[e,u]}= 1$ if vertex $u$ is the  head of edge $e$, $\BB_{[e,u]} = -1$ if $u$ is tail $e$, and $\BB_{[e,u]} = 0$ otherwise. Let  $\ee_u$  denote the $u^{\rm th}$ standard basis vector of appropriate dimension.  For an oriented edge $e\in E$ with end vertices   $u$ and $v$, we define $\bb_e = \bb_{u,v}=\ee_u - \ee_v$ if $u$ and $v$ are, respectively, the head and tail of $e$. Let  $\WW_{m\times m}$ be a diagonal matrix with $\WW_{[e,e]} = w(e)$. Then the Laplacian matrix $\LL$ of $G$ can also be written as $\LL = \BB^\top \WW \BB$. It is easy to verify that $\LL = \sum\nolimits_{e\in E} w(e)\bb_e \bb_e^\top$, which means that $\LL$ is  positive semidefinite.
We refer to $w(e)\bb_e \bb_e^\top$ as the Laplacian of edge $e$.

The pseudoinverse $\LL^\dag$ of Laplacian matrix $\LL$ is $(\LL +\frac{1}{n}\JJ)^{-1} - \frac{1}{n}\JJ$~\cite{GhBoSa08}, where $\JJ$ is the matrix of  appropriate dimension with all entries being ones. Let $0 = \lambda_1 < \lambda_2 \leq \ldots \leq\lambda_{n}$ be eigenvalues of $\LL$ of a connected graph $G$. 
The nonzero minimum eigenvalue $\lambda_2$ and the maximum  eigenvalue  $\lambda_n$ are $\lambda_{2} \geq w_{\min }/ n^2 $~\cite{LiSc18} and $\lambda_{n} \leq w_{\max}n$~\cite{SpSr11}, respectively. 
Let  $\DD$ and $\HH$ be Laplacians of two connected graphs with the same vertex set.  It is easy to verify that if $\DD \preceq \HH$, then $\HH^\dag \preceq \DD^\dag$. In addition, although  Laplacian matrix $\LL$ is positive semidefinite, its  principal submatrices are positive definite. For any nonempty set $Q \subseteq V$, the minimum eigenvalue $\lambda_{\min }\left(\boldsymbol{L}_{Q}\right)$ and the maximum eigenvalue $\lambda_{\max }\left(\boldsymbol{L}_{Q}\right)$ of matrix $\boldsymbol{L}_{Q}$ satisfy $\lambda_{\min }\left(\boldsymbol{L}_{Q}\right) \geq\lambda_2\geq w_{\min } / n^{2}$ and $\lambda_{\max }\left(\boldsymbol{L}_{Q}\right) \leq \lambda_n\leq n^{2} w_{\max }$, respectively~\cite{PaBa10}.

\begin{lemma}~\cite{MaZhWa16}\label{fact:inv+}
	Let $\LL$ be the Laplacian of a connected graph and 	let $\XX$ be a nonnegative, diagonal matrix with at least	one nonzero entry. Then, 	$\LL + \XX$ is positive definite, 	and every entry of $\kh{\LL + \XX}^{-1}$ is positive.
\end{lemma}

\subsection{Electrical Network and Effective Resistance}

For a connected graph $G=(V,E,w)$, we can define a corresponding electrical network $G=(V,E,r)$ by considering edges as resistors and considering vertices as junctions between resistors. The resistor of an associated  edge $e$ is $r_e=1/w(e)$.  The resistance distance $\mathcal{R}(u,v)$ between two vertices $u$ and $v$ in graph  $G$  is defined as the effective resistance between $u$ and $v$ in  the corresponding electrical network. Specifically, $\mathcal{R}(u,v)$ is equal  to the potential difference between $u$ and $v$ when a unit current is injected to $u$ and extracted from $v$.    For a  vertex $v$, its effective resistance  denoted by $\mathcal{R}_v$ is the sum of $\mathcal{R}(u,v)$ over all vertices in $G$, i.e., $\mathcal{R}_v=\sum_{u \in V }\mathcal{R}(u,v)$.  It has been shown that the resistance distance $\mathcal{R}(u,v)$  is equal to diagonal entry  $({\LL_{\{v\}}^{-1}})_{[u,u]}$ of $\LL_{\{u\}}$~\cite{INK+13}. Then, we have  $\mathcal{R}_v = \trace{{\LL_{\{v\}}^{-1}}}$.



In addition to the effective  resistance between two vertices, one can define effective  resistance $\mathcal{R}(u,Q)$ between vertex $u$ and a group $Q$ of vertices~\cite{ClPo11}.  Define $\Omega_{u,Q}$ as the current exiting vertex  $u$ when the vertices in $Q$  are grounded (i.e. have voltage 0) and  vertex  $u$ has voltage
1.  Then the effective resistance $\mathcal{R}(u,Q)$  is equal to reciprocal of  $\Omega_{u,Q}$, that is $\mathcal{R}(u,Q)=1/\Omega_{u,Q}$.
By definition, $\mathcal{R}(u,Q)=0$ if $u \in Q$.  The  effective resistance $\mathcal{R}_Q$ of a vertex set  $Q$ is defined as  the sum of $\mathcal{R}(u,Q)$ over all vertices in $V$: $\mathcal{R}_Q=\sum_{u\in V}\mathcal{R}(u,Q)$.  When set $Q$ includes only one vertex $v$, we have $\mathcal{R}(u,Q)=\mathcal{R}(u,v)$ and $\mathcal{R}_Q=\mathcal{R}_v$.  It was shown~\cite{ClPo11} that $\mathcal{R}(u,Q)=(\LL^{-1}_{Q})_{[u,u]}$ for  $u \notin Q$ and  $\mathcal{R}_Q={\rm Tr}(\LL_Q^{-1})$. 

\section{Problem Formulation} 

In this section, we first introduce the leader-follower continuous-time Degroot opinion dynamics model with noise~\cite{PaBa10}. Then we formulate the optimization problem and analyze the properties of objective function.

\subsection{Noisy Leader-Follower Opinion Dynamics}

For the leader-follower model on a social network $G=(V,E,w)$ with $n=|V|$ nodes and $m=|E|$ edges, where the $n$ nodes represent agents and the $m$ edges denote social affinity between agents. Moreover, nodes are classified into two groups: a set $Q$ of $q$ nodes are leader, and the remaining $n-q$ nodes in set $ V\backslash Q$ are followers. In leader-follower model, each node $i$ has a real-valued opinion $\xx_i(t) $ at time $t$.  
If $i$ is a leader node, its opinion remains unchanged over time, with $\xx_i(t)$ being a constant $\bar{x}$ for all $t$. If $i$ is a follower node, its opinion is updated based on only its own opinion and the opinions of its neighbors, and is subject to stochastic disturbances at the same time. Concretely, the opinion dynamics of each follower node $i$ is described by
\begin{align}\label{xdotv.eq}
\dot{\xx}_i(t)&=-\sum_{j \sim i} \WW_{[i,j]} [\xx_i(t)- \xx_ j(t)]+\zeta_i(t)\,,
\end{align}
where $\zeta_i(t)$ is a white noise with zero-mean and unit variance.

In the leader-follower opinion model, due to the presence of noise, the opinions of followers do not converge, but fluctuate around the opinion $\bar{x}$ of leaders,  in spite that the objective is for all nodes to follow the opinion $\bar{x}$. To measure how the opinion of each follower derivates from the opinion $\bar{\xx}$ of leaders in network $G$, we introduce the notion of polarization~\cite{MuMuTs18,XuBaZh21} denoted by $P_Q(G)$, which equals the steady-state variance of the deviation from consensus value $\bar{\xx}$ of the system,
\begin{equation}\label{polarization}
P_Q(G):=\lim _{t \rightarrow \infty} \sum_{i \in V} \mathbf{E}\left\{\left(\xx_{i}(t)-\bar{\xx}\right)^{2}\right\}.
\end{equation}
The smaller the $P_Q(G)$, the less the influence of noise on the leader-follower opinion dynamics, and the vice versa. Without loss of generality, in this paper we assume that $\bar{\xx}= 0$. In this case, $P_Q(G)$ is one-half the effective resistance $\mathcal{R}_Q$ of a node set $Q$ of leaders.
\begin{lemma}\cite{PaBa10}
For the noisy leader-follower opinion dynamics model on graph  $G=(V,E,w)$, where a set $Q \subset V$ of nodes are leaders with fixed opinion $0$,  the polarization $P_Q(G)$ is
	\begin{equation}\label{polarization}
P_Q(G)=\frac{1}{2} \trace{\LL_Q^{-1}}=\frac{1}{2}\mathcal{R}_Q\,. 
	\end{equation}
\end{lemma}

\subsection{Problem Statement}

As shown in~\eqref{polarization}, for the noisy leader-follower opinion dynamics in graph $G$ with a given group $Q$ of $q$ leader nodes described in~\eqref{xdotv.eq}, the quantity polarization $P_Q(G)$ is closely related to the position of leader nodes, implying that the selection of leaders has a strong impact on this   quantity, which can be used to measure the performance or role  of the group $Q$ of nodes in the opinion dynamics: the smaller the quantity $P_Q(G)$, the better the performance of the node group $Q$. Due to the equivalence of $\mathcal{R}_Q$ or $P_Q(G)$, below we alternatively use $\mathcal{R}_Q$, instead of $P_Q(G)$, to denote the performance of the nodes in group $Q$.

As will be shown later, if we add some edges, for each of which one end node is in $Q$ and the other end node is in $V\setminus Q$, then the effective resistance $\mathcal{R}_Q$ of the node group $Q$ will decrease. Let $E_Q$ be the set of nonexistent edges, each having a given weight and being incident to nodes in $Q$ and $V\setminus Q$, respectively. Then, the following problem arises naturally: How to optimally select a subset $S$ of $E_Q$, which includes $k$ edges, so that the effective resistance of node set $Q$ is minimized. In the sequel, we will address this optimization problem, which can be stated in a formal way as follows.
\begin{problem}\label{PRB1}
	Given a connected graph $G=(V,E,w)$, a node set $Q$, a candidate edge set $E_Q$ and an integer $k$, find an edge set $S$ satisfying $S \subset E_Q$ and $|S|=k$ such that effective resistance of node group $Q$ is minimized.
\end{problem}

Let $G^\prime $ denote the network augmented by adding the edges in $S$ to $G$, i.e. $G^\prime = (V,E\cup S, w^\prime)$, where $w^\prime: E\cup S \to \rea_{+}$ is the new weight function. Let $\mathcal{R}_Q(S)$ denote the effective resistance of the node group $Q$ in the augmented network $G^\prime(V,E \cup S,w')$, and let $\LL(S)$ be the Laplacian matrix  of $G^\prime$. Then, the set function optimization problem~\eqref{PRB1} can be formulated as:
\begin{align}\label{pro2}
\underset{S\subset E_Q,\, \sizeof{S}=k}{\mathrm{arg\,min}} \quad \mathcal{R}_Q(S)= \trace{\LL(S)_Q^{-1}}.
\end{align}

To solve problem~\eqref{PRB1}, exhaustively searching for the set $S$ of $k$ edges that maximally decreases the quantity $\mathcal{R}_Q(S)$ of $Q$ needs to calculate this quantity for every possible combination of $k$ edges out of the set  $E_Q$. Intuitively, this constitutes a combinatorial optimization problem with an exponential computational complexity.

\subsection{Properties  of Objective Function} 


Next we show that the objective function in~\eqref{pro2} has two desirable properties, that is, it is monotone and supermodular. Let $2^{E_Q}$ represent all the subsets of $E_Q$. Then the effective resistance of vertex group $Q$  in the augmented graph can be represented as a set function $\mathcal{R}_Q(\cdot): 2^{E_Q} \rightarrow \mathbb{R}$.  We first prove that function  $\mathcal{R}_Q(\cdot)$ is monotone.

\begin{theorem}\label{thm:MI}
	$\mathcal{R}_Q(S)$ is a monotonically decreasing function of the edge set  $S$. In other words, for any subsets $S \subset T \subset E_Q$, $\mathcal{R}_Q(T) < \mathcal{R}_Q(S)$ holds.
\end{theorem}
\begin{theorem}\label{thm:SM}
	$\mathcal{R}_Q(S)$ is supermodular. That is,  for any set $S \subset T \subset E_Q$ and any edge $e \in E_Q \setminus T$, the relation $\mathcal{R}_Q(T) - \mathcal{R}_Q(T \cup \{ e \} ) \leq \mathcal{R}_Q(S) - \mathcal{R}_Q(S \cup \{ e \} )$ holds.
\end{theorem}
\section{Simple Greedy Algorithm}

Since the set function $\mathcal{R}_Q(S)$ is  monotone and supermodular, the optimization problem in~\eqref{pro2} can be approximately solved by a simple greedy algorithm  with a provable optimality bound~\cite{NeWoFi78}.  Initially,  the augmented edge set $S$ is set to be empty. Then $k$ edges are iteratively selected to the augmented edge set from set $E_Q\setminus S$. In each iteration  of the greedy algorithm, the edge $e$ in the candidate set is chosen to maximize the quantity $\mathcal{R}_Q(S)-\mathcal{R}_Q\kh{S \cup \{ e \} }$. The algorithm stops when $k$ edges are selected to be added to $S$.  The computation for $\mathcal{R}_Q(S)$ can be performed according to the following lemma.

\begin{lemma} \label{lem:SMF}
	For a connected weighted graph $G=(V,E,w)$ with a set $Q$ of $q$ target vertices, weighted Laplacian matrix $\LL$, let $e$ be a nonexistent  edge with given weight $w(e)$ connecting two vertices  $v \in V \setminus Q$ and  $u \in Q$,  and let $S$ be the set of added edges. Let $\LL(S)$ denote the Laplacian matrix of the augmented graph. Then,
	\begin{align*}
		\LL(S)_Q=\LL_Q+\sum\limits_{(u,v) \in S} w(u, v)\EE_{uu},
	\end{align*}
	where $\EE_{uu}=\ee_u \ee_u^\top$.
\end{lemma}

Considering $\mathcal{R}_Q(S) ={\rm Tr}(\LL(S)_Q^{-1})$, a na\"{\i}ve greedy algorithm requires $O(k|E_Q|  (n-q)^3)$ time, which is computationally intractable even for  small-size networks. Below we show that the computation time can be  greatly reduced.

At each iteration of the greedy algorithm, only the edge $e$ with maximum $\mathcal{R}_Q(S)-\mathcal{R}_Q\kh{S \cup \{ e \} }$, denoted by $\mathcal{R}_Q^\Delta(e)$,  is chosen. In the  na\"{\i}ve  algorithm, one needs to compute  $\LL(S)_Q^{-1}$ after each update of $S$ in  $O((n-q)^3)$ time. Actually, in each iteration, the new matrix $\LL(S\cup\{e\})_Q$ is a rank one perturbation of the matrix $\LL(S)_Q$, that is, $\LL(S\cup\{e\})_Q^{-1} = \kh{\LL(S)_Q+ w(e)\ee_u\ee_u^\top}^{-1}$. Then, exploiting Sherman-Morrison formula~\cite{Me73}, $\LL(S\cup\{e\})_Q^{-1}$ can be found  by applying a rank one update to $\LL(S)_Q^{-1}$ in time $O((n-q)^2)$, rather than directly  computing the inverse of matrix $\LL(S\cup\{e\})_Q^{-1}$ that takes time $O((n-q)^3)$.  Therefore, the quantity $\mathcal{R}_Q^\Delta(e)$  can be evaluated as
\begin{align}\label{eq:SM}
	&\quad \mathcal{R}_Q^\Delta(e) =  \mathcal{R}(S)-\mathcal{R}(S\cup\{e\})=\LL(S)_Q^{-1}-\LL(S\cup\{e\})_Q^{-1} \notag\\
	&=\LL(S)_Q^{-1}-\left(\LL(S)_Q^{-1}- \frac{w(e) \LL(S)_Q^{-1} \ee_u \ee_u^\top \LL(S)_Q^{-1}}{1 + w(e)\ee_u^\top \LL(S)_Q^{-1}\ee_u} \right) \notag\\
	&=\frac{w(e) \LL(S)_Q^{-1} \ee_u \ee_u^\top \LL(S)_Q^{-1}}{1 + w(e)\ee_u^\top \LL(S)_Q^{-1}\ee_u}
	= \frac{w(e) \norm{\LL(S)_Q^{-1} \ee_u}^2  }{1 + w(e)\ee_u^\top \LL(S)_Q^{-1} \ee_u}.
\end{align}

Equation~\eqref{eq:SM} leads to Algorithm~\ref{alg:MVC1},  $\ExactSM(G, v, E_v, k)$. This algorithm  first computes the inverse of $\LL_{Q}$  in time $O((n-q)^3)$. Then it works in $k$ rounds, with each round mainly including two steps: computing $\mathcal{R}_Q^\Delta(e)$ (Lines 4-8)  in $O(n(n-q))$ time, and updating  $\LL(S)_Q^{-1}$ (Line 13) in $O((n-q)^2)$ time. Thus, the total running time of Algorithm~\ref{alg:MVC1} is $O((n-q)^3+kn(n-q))$, which is much   faster than the na\"{\i}ve   algorithm.

\begin{algorithm}[htbp]
	
	\caption{$\ExactSM(G, Q, E_Q, k)$}
	\label{alg:MVC1}
	\Input{
		A connected graph $G$; a node set $Q \subset V$; a candidate edge set $E_Q$;
		an integer $k \leq |E_Q|$
	}
	\Output{
		$S$: a subset of $E_Q$ and $|S| = k$
	}
	Compute $\LL_Q^{-1}$ \;
	$S = \emptyset $ \;
	\For{$i = 1$ to $k$}{
		$t(u) = \norm{\LL_Q^{-1} \ee_u}^2$ for all $u \in V \setminus Q$\;
		$r(u) = \ee_u^\top \LL_Q^{-1} \ee_u$ for all $u \in V \setminus Q$\;
		\For{ each $e \in E_Q\setminus S$}
		{
			$u=$ the vertex that $e$ connects in set $V \setminus Q $\;
			$\mathcal{R}_Q^\Delta(e)=
			\frac{w(e) t(u) }{1 + w(e) r(u)}$}
		$e_i = \mathrm{arg\,max}_{e \in E_Q \setminus S} \mathcal{R}_Q^\Delta(e)$ \;
		$S = S \cup \{ e_i \}$ \;
		$G = G(V,E \cup \{ e_i \})$ \;
		$u=$ the vertex that $e_i$ connects in set $V \setminus Q $\;
		$\LL_Q^{-1}=\LL_Q^{-1}- \frac{w(e_i) \LL_Q^{-1} \ee_u \ee_u^\top \LL_Q^{-1}}{1 + w(e_i)\ee_u^\top \LL_Q^{-1} \ee_u}$
	}
	\Return $S$
\end{algorithm}

Based on the well-established   result~\cite{NeWoFi78}, Algorithm~\ref{alg:MVC1} yields a  $(1-1/e)$-approximation of the optimal solution to the  problem in~\eqref{pro2}, as provided in the following theorem.
\begin{theorem}
	The edge set $S$ returned by Algorithm~\ref{alg:MVC1} satisfies relation
	\begin{align*}
		\mathcal{R}_Q(\emptyset) - \mathcal{R}_Q(S) \geq (1-1/e)(\mathcal{R}_Q(\emptyset) - \mathcal{R}_Q(S^*)),
	\end{align*}
	where $S^*$ is the optimal solution to~(\ref{pro2}), that is,
	\begin{align*}
		S^* \, \, \defeq \underset{S\subset E_Q,\, \sizeof{S}=k}{\mathrm{arg\,min}} \quad \mathcal{R}_Q(S) .
	\end{align*}
\end{theorem}
\section{Fast Greedy Algorithm}

Although the computation time of Algorithm~\ref{alg:MVC1} is significantly reduced, compared with the na\"{\i}ve algorithm, it is still  computationally unacceptable for large networks with millions of vertices, since it requires computing the inverse of  matrix $\LL_Q$. Below we  present  an efficient  approximation algorithm, which avoids inverting the matrix $\LL_Q$ but returns a  $(1 - 1/e-\eps)$ approximation  of the optimal solution to  problem~(\ref{pro2}) in time $O(m \eps^{-2} \log^{2.5} n \log \eps^{-1} \operatorname{polyloglog}(n)+qn)$.

The key step for  solving the  problem in~(\ref{pro2}) is to compute the quantity  $\mathcal{R}_Q^\Delta(e)$. According to~\eqref{eq:SM}, to evaluate  $\mathcal{R}_Q^\Delta(e)$, one needs to estimate the  two terms $\norm{\LL(S)_Q^{-1} \ee_u}^2$ and $\ee_u^\top \LL(S)_Q^{-1} \ee_u$ in numerator and denominator, respectively. In the next two subsections, we provide efficient approximations for these two quantities  $\mathcal{R}_Q^\Delta(e)$.


\subsection{Approximating the Norm in~\eqref{eq:SM}}

We first approximate $\norm{\LL(S)_Q^{-1} \ee_u}^2$. It is the $\ell_2$ norm of  a vector $\LL(S)_Q^{-1} \ee_u$ in $\mathbb{R}^{n-q}$.  However, the complexity for exactly computing this  $\ell_2$ norm  is high. To reduce the computation cost, we will  apply the Johnson-Lindenstrauss (JL) lemma~\cite{JoLi84,Ac03},  which nearly  preserves  the $\ell_2$ norm by projecting the vector $\LL(S)_Q^{-1} \ee_u$ onto a low-dimensional subspace, but significantly reduces the computational cost. For consistency, we  introduce the JL lemma~\cite{JoLi84,Ac03}.
\begin{lemma}
	\label{lem:jl}
	Let $\vv_1,\vv_2,\cdots,\vv_n \in \mathbb{R}^d$ be $n$ fixed $d-$dimensional  vectors and $ \eps > 0$ be a real number.
	Let $p$ be a positive integer such that $p \geq 24\log n/\eps^2$ and let $\QQ_{p \times d}$ be a random matrix with each entry being $1/\sqrt{p}$ or $-1/\sqrt{p}$ with identical probability. Then, with probability at least $1 - 1/n$, the following statement holds for any pair of $i$ and $j$, $1\leq i, j\leq n$:
	\[
	\norm{\QQ \vv_i - \QQ \vv_j}^2 \approx_{\eps} \norm{\vv_i - \vv_j}^2.
	\]
\end{lemma}

Let $\PP_{p\times n}$ be a random $\pm1/\sqrt{p}$ matrix with $p =\ceil{24\log n/ \eps^2}$. By Lemma~\ref{lem:jl}, we have 
\begin{align}
	\norm{\PP \LL(S)_Q^{-1} \ee_u}^2 \approx_{\eps} \norm{\LL(S)_Q^{-1} \ee_u}^2
\end{align}
holds for any $ u\in V$ with high probability.  However, if we directly compute $\norm{\PP \LL(S)_Q^{-1} \ee_u}^2$,  we should invert  $\LL(S)_Q$, which is time-consuming. In order to avoid matrix inverse, we use the fast symmetric, diagonally dominant (SDD) linear system solvers~\cite{SpTe14,CoKyMiPaJaPeRaXu14} to compute $\PP \LL(S)_Q^{-1}$.
\begin{lemma}\label{lem:solve}
	There is a nearly linear time solver $\xx = \SDDMSolver(\SS, \bb, \delta)$ which takes an SDD matrix $\SS_{n\times n}$ with $m$ nonzero entries, a vector $\bb \in \mathbb{R}^n$, and an error parameter $\delta > 0$, and returns a vector $\xx \in \mathbb{R}^n$ satisfying $\norm{\xx - \SS^{-1} \bb}_{\SS} \leq \delta \norm{\SS^{-1} \bb}_{\SS}$ with high probability, where $\norm{\xx}_{\SS} \defeq \sqrt{\xx^\top \SS \xx}$. The solver runs in expected time $O(m\log^{0.5}n \log\delta^{-1} \operatorname{polyloglog}(n))$ .
\end{lemma}

Using Lemmas \ref{lem:solve} and \ref{lem:jl}, $\PP \LL(S)_Q^{-1}$ can be efficiently approximated as stated in the following lemma.
\begin{lemma}\label{lem:z1-zhat}
	Let $\ZZ^{(1)} = \PP \LL(S)_Q^{-1}\ee_u$. If there exists a matrix $\ZZtil^{(1)}$ satisfying $\ZZtil^{(1)}_{[i,:]} = \SDDMSolver(\LL(S)_Q, \PP_{[i,:]}, \delta_1 )$, where
	\begin{align}
		\delta_1 \leq \frac{\eps \sqrt{1-\eps}w_{\rm min}}{6n^3w_{\rm max}}\label{DLT1}
	\end{align}
	Then,
	\begin{align*}
		\norm{\ZZ^{(1)} \ee_u}^2 \approx_{\eps} \norm{\ZZtil^{(1)} \ee_u}^2
	\end{align*}
	holds for any $ u\in V$ with high probability.
\end{lemma}

\subsection{Approximating the Quadratic Form in~\eqref{eq:SM}}

We continue to approximate $\ee_u^\top \LL(S)_Q^{-1}\ee_u$.  Note that  $\LL(S)_Q$ is an SDD  matrix and  can be expressed   in terms of the  sum of a Laplacian $\BB'^\top \WW' \BB'$ and a nonnegative diagonal matrix $\XX$  as $\LL(S)_Q = \BB'^\top \WW' \BB' + \XX$. Then,
\begin{align}\label{eq:deno}
	&\ee_u^\top \LL(S)_Q^{-1} \ee_u
	= \ee_u^\top \LL(S)_Q^{-1} \LL(S)_Q \LL(S)_Q^{-1} \ee_u \notag \\
	=& \ee_u^\top \LL(S)_Q^{-1} \kh{\BB'^\top \WW' \BB' + \XX}
	\LL(S)_Q^{-1} \ee_u \notag \\
	=&
	\ee_u^\top \LL(S)_Q^{-1} \BB'^\top \WW' \BB'
	\LL(S)_Q^{-1} \ee_u +  \ee_u^\top \LL(S)_Q^{-1} \XX
	\LL(S)_Q^{-1} \ee_u \notag\\
	=&
	\norm{\WW'^{1/2} \BB' \LL(S)_Q^{-1} \ee_u}^2 +
	\norm{\XX^{1/2} \LL(S)_Q^{-1} \ee_u}^2.
\end{align}
Thus,  the determination of  $\ee_u^\top \LL(S)_Q^{-1}\ee_u$ can be  reduced to evaluating $\ell_2$ norm of  vectors  in $\mathbb{R}^m$ and $\mathbb{R}^n$.  Let $\QQ_{p\times m}$ and $\RR_{p\times n}$ be two random $\pm1/\sqrt{p}$ matrices where $p = \ceil{24\log n/ \eps^2}$. By Lemma~\ref{lem:jl}, for any $ u\in V$ we have
\begin{align*}
	 \ee_u^\top \LL(S)_Q^{-1} \ee_u 
	\approx_{\eps} 
	\norm{\QQ \WW'^{1/2} \BB' \LL(S)_Q^{-1} \ee_u}_2^2 +
	\norm{\RR \XX^{1/2} \LL(S)_Q^{-1} \ee_u}_2^2.
\end{align*}

Applying Lemmas~ \ref{lem:jl}  and~\ref{lem:solve}, the above two formulas can be approximated as follows.
\begin{lemma}\label{lem:z3-zhat}
	Let $\ZZ^{(2)}=\QQ \WW'^{1/2} \BB'\LL_Q^{-1}$, $\ZZ^{(3)}= \RR \XX^{1/2} \LL_Q^{-1}$. If there are two matrices $\ZZtil^{(2)}$ and $\ZZtil^{(3)}$ satisfying
	\begin{align*}
		\ZZtil^{(2)}_{[i,:]} =& \SDDMSolver\kh{\LL(S)_Q,(\QQ \WW'^{1/2} \BB')_{[i,:]}, \delta_2} \quad \rm and \\
		\ZZtil^{(3)}_{[i,:]} =& \SDDMSolver\kh{\LL(S)_Q,(\RR \XX^{1/2})_{[i,:]}, \delta_2}, \quad \rm where
	\end{align*}
	\begin{align}
		\delta_2 \leq \sqrt{\frac{\eps w_{\rm min}^2}{16n^5m^2} \sqrt{\frac{2-2\eps}{w_{\rm max}}}}. \label{DLT2}
	\end{align}
	Then,  for any $u\in V$, the following relation holds:
	\begin{align*}
		\norm{\ZZ^{(2)} \ee_u}^2 + \norm{\ZZ^{(3)} \ee_u}^2 \approx_{\eps}\norm{\ZZtil^{(2)} \ee_u}^2 + \norm{\ZZtil^{(3)} \ee_u}^2.
	\end{align*}
\end{lemma}

\subsection{Fast Algorithm for Approximating $\mathcal{R}_Q^\Delta(e)$}

Based on Lemmas~\ref{lem:z1-zhat} and~\ref{lem:z3-zhat}, we propose an algorithm $\GainsEst$ approximating $\mathcal{R}_Q^\Delta(e)$ of every edge in the candidate set $E_Q$. The outline of  algorithm  $\GainsEst$ is shown in Algorithm \ref{alg:MVC2}, and its performance is given in Theorem \ref{lem:edgebyjl}.

\begin{algorithm}[htbp]
	\caption{$\GainsEst\kh{G,\LL(S)_Q, Q, E_Q, \eps}$}
	\label{alg:MVC2}
	\Input{
		A graph $G$; a sparse matrix $\LL_Q$; a node set $Q \subset V$; a candidate edge set $E_Q$; a real number $0 \leq \epsilon \leq 1/4$
	}
	\Output{
		$\{ (e,\hat{\mathcal{R}}_Q^\Delta(e)) | e \in E_Q \}$
	}
	
	set $\delta_1$ and $\delta_2$ according to (\ref{DLT1}) and (\ref{DLT2}) \;
	$p =\ceil{24\log n/ \eps^2}$ \;
	Compute sparse matrices $\WW'^{1/2}$, $\BB'$ and $\XX^{1/2}$\;
	Construct three random $\pm1/\sqrt{p}$ matrices $\PP_{p\times n}$, $\QQ_{p\times m}$ and $\RR_{p\times n}$\;
	$\boldsymbol{\mathit{Y}}^{(1)}=\QQ \WW'^{1/2} \BB'$ , $\boldsymbol{\mathit{Y}}^{(2)}=\RR \XX^{1/2}$\;
	\For{$i = 1$ to $p$}{
		$\ZZtil^{(1)}_{[i,:]} = \SDDMSolver(\LL(S)_Q,\PP_{[i,:]}, \delta_1)$\;
		$\ZZtil^{(2)}_{[i,:]} = \SDDMSolver(\LL(S)_Q,\boldsymbol{\mathit{Y}}^{(1)}_{[i,:]}, \delta_2)$\;
		$\ZZtil^{(3)}_{[i,:]} = \SDDMSolver(\LL(S)_Q,\boldsymbol{\mathit{Y}}^{(2)}_{[i,:]}, \delta_2)$
	}
	\For{each $e\in E_Q$}{
		$u=$ the vertex that $e$ connects in set $V \setminus Q $\;
		$\hat{\mathcal{R}}_Q^\Delta(e) = \frac{w(e)\norm{\ZZtil^{(1)}\ee_u}^2}{1+w(e)\left({\norm{\ZZtil^{(2)}\ee_u}^2 + \norm{\ZZtil^{(3)}\ee_u}^2}\right)}$
	}
	\Return $\{ (e,\hat{\mathcal{R}}_Q^\Delta(e)) | e \in E_Q \}$.
\end{algorithm}

\begin{theorem}
	\label{lem:edgebyjl}
	Given a connected undirected graph $G = (V,E,w)$ with $n$ vertices, $m$ edges, positive edge weights
	$w : E \to \mathbb{R}_{+}$, a set $Q\subseteq V$ of $q$ target vertices, a set $E_Q$ of edges, each connecting one vertex in $Q\subseteq V$ and one vertex in $V  \setminus Q$,
	and scalars $0<\epsilon \leq1/4$,
	the algorithm $\GainsEst$ returns a set of pairs $\{(e_u, \hat{\mathcal{R}}_Q^\Delta(e_u)) | e_u \in E_Q\}$.
	With high probability, the following statement holds for $\forall e \in E_Q$
	\begin{align}\label{eq:del}
		\mathcal{R}_Q^\Delta(e) \approx_{3\epsilon} \hat{\mathcal{R}}_Q^\Delta(e).
	\end{align}
	The total running time of this algorithm is bounded by \\
	$O(m \eps^{-2} \log^{2.5} n \log \eps^{-1} \operatorname{polyloglog}(n)+q(n-q)\eps^{-2} \log n)$.
\end{theorem}

\subsection{Accelerated  Algorithm for Approximating $\mathcal{R}_Q^\Delta(e)$}

Although Algorithm~\ref{alg:MVC2} is fast, it can still be improved both in the space and runtime requirements. In Algorithm~\ref{alg:MVC2},  intermediate variables $\PP$, $\RR$ and $\QQ$   are stored in three matrices: two $p \times n$ matrices and one $p \times m$ matrix. In fact, these intermediate variables can be replaced by three vectors: two $1 \times n$ vectors $\pp^\top$ and $\rr^\top$, and one $1 \times m$ verctor $\qq^\top$. Analogously, intermediate variables $\ZZtil^{(1)}$, $\ZZtil^{(2)}$ and $\ZZtil^{(3)}$ in Algorithm~\ref{alg:MVC2} need not be stored in three $p \times n$ matrices but, instead, three vectors $\zztil_1^\top$, $\zztil_2^\top$ and $\zztil_3^\top$ of size $1 \times n$ are sufficient. These observations result in significant improvement in the space requirement of Algorithm~\ref{alg:MVC2}, based on which we propose an accelerated algoirhtm $\FGainsEst$ shown in Algorithm~\ref{alg:FMVC2}.

In addition to the space-efficient implementation, Algorithm~\ref{alg:FMVC2} also reduces the computation cost, in contrast with Algorithm~\ref{alg:MVC2}. Note that in Algorithm~\ref{alg:MVC2}, the execution of Line 12 is decomposed into two parts: one is evaluating numerator, the other is computing denominator. In order to obtain $\hat{\mathcal{R}}_Q^\Delta(e)$, each part needs to compute the $\ell_2$ norm in $p = O(\eps^{-2} \log{n})$ time, leading to the $q(n-q)\eps^{-2} \log n$ complexity for the second loop. In contrast, in Algorithm~\ref{alg:FMVC2}, three random vectors $\pp^\top$, $\rr^\top$, and   $\qq^\top$ are created and exploited for projecting the nodes. In this case, the two parts of $\hat{\mathcal{R}}_Q^\Delta(e)$ are computed additively. Specifically, in each iteration of the first  loop, the contributions $\hat{t}(u)$ and $\hat{r}(u)$ to $\hat{\mathcal{R}}_Q^\Delta(e)$ are updated by adding related quantities (Line 10 and Line 11 of Algorithm~\ref{alg:FMVC2}). Since $\hat{t}(u)$ and $\hat{r}(u)$ are computed in the first loop, the cost of the second loop of Algorithm~\ref{alg:MVC2} is reduced to $O(q(n-q))$. Thus, Algorithm~\ref{alg:FMVC2} runs in $O(m \eps^{-2} \log^{2.5} n \log \eps^{-1} \operatorname{polyloglog}(n)+q(n-q))$.

On the other hand,  the computation time of Algorithm~\ref{alg:FMVC2} can be further reduced since it is amenable to parallel implementation. Specifically, for the first \textbf{for} loop,  each iteration  can be executed  independently and in parallel, in different cores. The result of parallel treatment does not affect the solution returned by the algorithm,  but  leads to significant improvement in running time: in a parallel system with $O(\log{n})$ cores, the running time of the parallel version of Algorithm~\ref{alg:FMVC2} is $O(m \eps^{-2} \log^{1.5} n \log \eps^{-1} \operatorname{polyloglog}(n)+q(n-q))$. In all our experiments,  this parallelization is  used to reduce the running time.

\begin{algorithm}[htbp]
	\caption{$\FGainsEst\kh{G,\LL_Q, Q, E_Q, \eps}$}
	\label{alg:FMVC2}
	\Input{
		A graph $G$; a sparse matrix $\LL_Q$; a node set $Q \subset V$; a candidate edge set $E_Q$; a real number $0 \leq \epsilon \leq 1/4$
	}
	\Output{
		$\{ (e,\hat{\mathcal{R}}_Q^\Delta(e)) | e \in E_Q \}$
	}
	
	set $\delta_1$ and $\delta_2$ according to (\ref{DLT1}) and (\ref{DLT2}) \;
	$p =\ceil{24\log n/ \eps^2}$ \;
	Compute sparse matrices $\WW'^{1/2}$, $\BB'$ and $\XX^{1/2}$\;
	$\hat{t}(u) = \hat{r}(u) = 0$ for all $u \in V \setminus Q$ \;
	\For{$i = 1$ to $p$}{
		Construct three $\pm 1 / \sqrt{p}$ random vectors $\pp^\top$, $\qq^\top$ and $\rr^\top$ \;
		$\zztil_1^\top
		= \SDDMSolver(\LL(S)_Q,
		\pp^\top, \delta_1)$ \;
		$\zztil_2^\top
		=\SDDMSolver(\LL(S)_Q,
		\qq^\top \WW'^{1/2} \BB', \delta_2)$\;
		$\zztil_3^\top
		= \SDDMSolver(\LL(S)_Q,
		\rr^\top \XX^{1/2}, \delta_2)$\;
		$\hat{t}(u) = \hat{t}(u)+\kh{\zztil_1^\top \ee_u}^2$ for all $u \in V \setminus Q$ \;
		$\hat{r}(u) = \hat{r}(u)+\kh{\zztil_2^\top \ee_u}^2+\kh{\zztil_3^\top \ee_u}^2$ for all $u \in V \setminus Q$ \;
	}
	\For{each $e\in E_Q$}{
		$u=$ the vertex that $e$ connects in set $V \setminus Q $\;
		$\hat{\mathcal{R}}_Q^\Delta(e) =
		\frac{w(e)\hat{t}(u)}
		{1+w(e)\hat{r}(u)}$
	}
	\Return $\{ (e,\hat{\mathcal{R}}_Q^\Delta(e)) | e \in E_Q \}$.
\end{algorithm}

\subsection{Fast Algorithm for Objective Function}

Exploiting Algorithm~\ref{alg:FMVC2} to approximate $\mathcal{R}_Q^\Delta(e)$, we propose a fast greedy algorithm $\ApproxiSM$ for solving problem~(\ref{pro2}), as reported in Algorithm~\ref{alg:MVC3}. The computational complexity of  algorithm $\ApproxiSM$ is easy to compute in the  following way. Note that Algorithm~\ref{alg:MVC3} iterates  $i$ times.
At each iteration $i$, it executes the call of $\FGainsEst$ in time $O(m \eps^{-2} \log^{2.5} n \log \eps^{-1} \operatorname{polyloglog}(n)+qn)$, finds edge $e_i$ in time $O(qn)$,  and performs  other operations in time $O(1)$. Thus, its total running time is $O(km \eps^{-2} \log^{2.5} n \log \eps^{-1} \operatorname{polyloglog}(n)+kq(n-q))$.

\begin{algorithm}[htbp]
	\caption{$\ApproxiSM(G, Q, E_Q, k, \epsilon)$}
	\label{alg:MVC3}
	\Input{
		A connected graph $G$; a node set $Q \subset V$; a candidate edge set $E_Q$;
		an integer $k \leq |E_Q|$; a real number $0 \leq \eps \leq 1/4$
	}
	\Output{
		$S$: a subset of $E_Q$ and $|S| = k$
	}
	$S = \emptyset$ \;
	Compute sparse matrix $\LL_Q$\;
	\For{$i = 1$ to $k$}{
		$\{(e,\hat{\mathcal{R}}_Q^\Delta(e)) \}=\FGainsEst(G,\LL_Q, Q, E_Q, \eps)$. \;
		$e_i = \mathrm{arg\,max}_{e \in E_Q\setminus S} \hat{\mathcal{R}}_Q^\Delta(e)$ \;
		$S = S \cup \{ e_i \}$ \;
		$G = G(V,E \cup \{ e_i \})$\;
		$u=$ the vertex that $e_i$ connects in set $V \setminus Q $\;
		$\LL_Q = \LL_Q + w(e_i)\EE_{uu}$
	}
	\Return $S$
	
\end{algorithm}

The output $\hat{S}$ of  Algorithm~\ref{alg:MVC3} gives a $\kh{1 - {1}/{e} - \eps}$ approximate solution to problem~(\ref{pro2}) as provided by the following theorem.
\begin{theorem}
	Given a connected undirected graph $G = (V,E,w)$ with $n$ vertices, $m$ edges, positive edge weights
	$w : E \to \mathbb{R}_{+}$, a set $Q\subseteq V$ of $q$ target vertices, a set $E_Q$ of edges, each connecting one vertex in $Q\subseteq V$ and one vertex in $V  \setminus Q$,
	and scalars $0<\epsilon \leq1/4$, 	the algorithm $\ApproxiSM(G, Q, E_Q, k, \epsilon)$ returns a set  $\hat{S}$ of $k$ edges in  $E_Q$,   satisfying
	\begin{align}\label{eq:bound1}
		\mathcal{R_Q}(\emptyset) - \mathcal{R_Q}(\hat{S})  \leq (1 - 1/e - \eps) (\mathcal{R_Q}(\emptyset) - \mathcal{R_Q}(S^*)),
	\end{align}
	where $S^*$ is the optimal solution to~(\ref{pro2}).
\end{theorem}
\begin{table}[tb]
	\centering
	\caption{Statistics of datasets for real-world networks  and the average running times  (seconds, $s$) of $\ExactSM$ and $\ApproxiSM$ algorithms on these networks. For any network,   $n$ and $m$  denote, respectively, the number of nodes and edges in its largest connected component. }\label{SetNo}
	\resizebox{1\columnwidth}{!}{
		\begin{tabular}{ccccc}
			\toprule
			\multirow{2}{*}{\centering Network} & \multirow{2}{*}{\centering $n$} & \multirow{2}{*}{\centering $m$} & \multicolumn{2}{c}{Running Time ($s$)} \\
			\cmidrule{4-5}
			\fontsize{6.5}{8}\selectfont
			& & & $\ExactSM$ &$\ApproxiSM$  \\
			\midrule
			Karate &  34 &  78 &  0.08 &  1.59 \\
			Windsurfers &  43 &  336 &  0.05 &  1.60 \\
			Dolphins &  62 &  159 &  0.09 &  1.68 \\
			Lesmis &  77 &  254 &  0.06 &  1.62 \\
			Adjnoun &  112 &  425 &  0.07 &  1.64 \\
			Celegansneural &  297 &  2148 &  0.13 &  1.78 \\
			Chicago &  823 &  822 &  0.47 &  1.84 \\
			Hamster Full &  2000 &  16098 &  1.19 &  4.00 \\
			Facebook &  4039 &  88234 &  4.17 &  36.15 \\
			GrQc &  4158 &  13422 &  4.14 &  4.95 \\
			Power Grid &  4941 &  6594 &  6.29 &  4.45 \\
			High Energy &  5835 &  13815 &  9.18 &  5.91 \\
			Reactome &  5973 &  145778 &  9.67 &  93.78 \\
			Route Views &  6474 &  12572 &  10.28 &  5.55 \\
			HepTh &  8638 &  24806 &  21.74 &  9.95 \\
			Pretty Good Privacy &  10680 &  24316 &  38.35 &  14.14 \\
			HepPh &  11204 &  117619 &  43.58 &  75.86 \\
			AstroPh &  17903 &  196972 &  161.98 &  201.64 \\
			Internet &  22963 &  48436 &  308.99 &  33.11 \\
			CAIDA &  26475 &  53381 &  447.77 &  39.31 \\
			Enron Email &  33696 &  180811 &  854.77 &  203.33 \\
			Condensed Matter &  36458 &  171735 &  1134.55 &  184.41 \\
			Brightkite &  56739 &  212945 &  3454.63 &  300.44 \\
			Word Net &  145145 &  656230 &  --- &  3883.64 \\
			Gowalla &  196591 &  950327 &  --- &  9947.73 \\
			DBLP &  317080 &  1049866 &  --- &  13476.64 \\
			Amazon &  334863 &  925872 &  --- &  10060.25 \\
			Pennsylvania &  1087562 &  1541514 &  --- &  42816.90 \\
			Texas &  1351137 &  1879201 &  --- &  65877.31 \\
			\bottomrule
		\end{tabular}
	}
\end{table}

\section{Experiments}

In this section, we experimentally evaluate the performance of our proposed greedy algorithms on some real-world  networks taken from KONECT~\cite{kunegis2013konect} and SNAP~\cite{LeSo16}. 
 For each network, we implement our experiments  on its largest connected components.   The information of the largest components for all networks is provided in Table~\ref{SetNo}, where  networks are listed in increasing size of the largest components.  The performance we evaluate  includes the quality of the solutions of both algorithms and their running time.

All algorithms in our experiments are executed in Julia. In our algorithms, we use the linear solver  $\SDDMSolver$~\cite{kyng2016approximate}. The source code of our algorithms is available at \url{https://github.com/vivian1tsui/optimize_polarization}. 
All experiments were conducted on a machine equipped with  32G RAM and 4.2 GHz Intel i7-7700 CPU. In our experiments,  the $q$  leader nodes in $Q$ are randomly selected from the set $V$ of all nodes. The candidate edge set $E_Q$ is composed of all nonexistent edges, each having  unit weight $w = 1$, with one end in $Q$ and the other end in $V\setminus Q$.  For the approximated algorithm $\ApproxiSM$,  we  set $\eps= 0.2$, since it is enough to achieve good performance.

\subsection{Effectiveness of Greedy Algorithms}

We first evaluate the effectiveness of our algorithms, by comparing them with both the optimum solutions and an alternative random scheme, by randomly selecting $k$ edges from $E_Q$.  For this purpose, we  execute   experiments on four small realistic networks: Karate network, Windsufers network, Dolphins network and Lesmis network. These networks are small, allowing us to compute the optimal set of edges. We consider two cases:  the cardinality  of $Q$ equals 3 or 5. For each case, we add  $k=1,2,\ldots,6$ edges, and the results reported are averages of 10 repetitions. Figures~\ref{smallQ3} and~\ref{smallQ5} report  the results for $|Q|=3$ and $|Q|=5$, respectively.   We observe that the  solutions returned by our two greedy algorithms
and the optimum solution are almost the same,  all of which are much better than those returned by the random scheme.

\begin{figure}[tb]
	\centering
	\includegraphics[width=\linewidth]{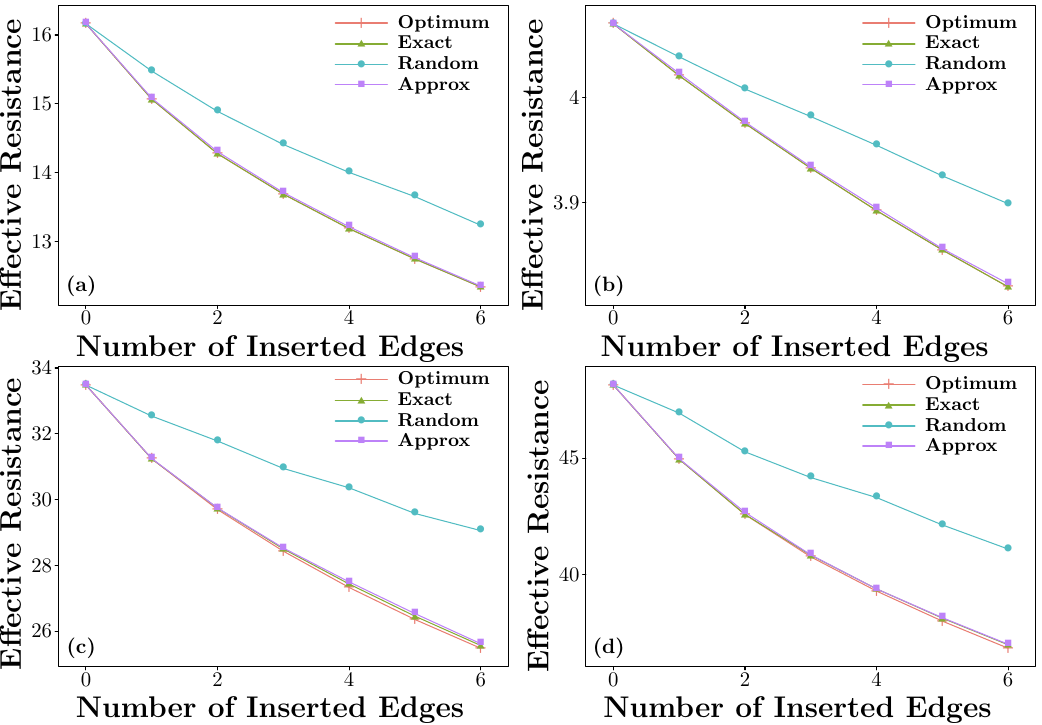}
	\caption{ Effective resistance for a set $Q$ of $q=3$ leader nodes   as a function of the number $k$ of inserted edges for $\ExactSM$, $\ApproxiSM$, random and the optimum solution on four networks: (a) Karate, (b) Windsufers, (c) Dolphins,  and (d) Lesmis.\label{smallQ3}}
\end{figure}

\begin{figure}[tb]
	\centering
	\includegraphics[width=\linewidth]{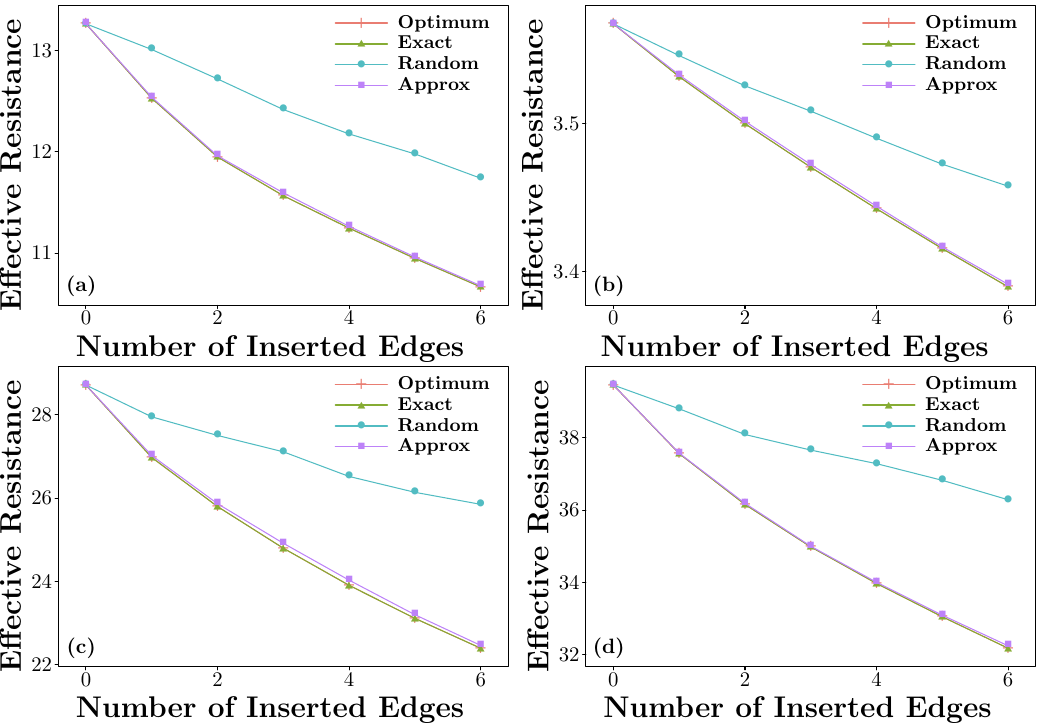}
	\caption{Effective resistance for a set $Q$ of $q=5$ leader nodes   as a function of the number $k$ of inserted edges for $\ExactSM$, $\ApproxiSM$, random and the optimum solution on four networks: (a) Karate, (b) Windsufers, (c) Dolphins,  and (d) Lesmis.\label{smallQ5}}
\end{figure}

To further show the accuracy of our algorithms,  we continue to compare our algorithms with some schemes  on four larger networks, including Chicago, Hamster Full, Facebook, and HepTh. Since these networks are large, we can hardly obtain the optimum solutions. We  consider the following three baselines,  random scheme,  TopDegree, and TopCent. In TopDegree (TopCent) scheme, we choose the node in $V\setminus Q$ with the highest degree (smallest effective resistance) and link it to random $k$ nodes in $Q$.  We also consider two cases:  $|Q|=5$ and $|Q|=10$.  In Figures~\ref{largeQ5} and~\ref{largeQ10}, we report the results for $k=1,2,\ldots,20$. Both figures show  that  there is little difference between the solutions of  our two greedy algorithms,  which are significantly  better than the solutions of the three baselines.

\begin{figure}[tb]
	\centering
	\includegraphics[width=\linewidth]{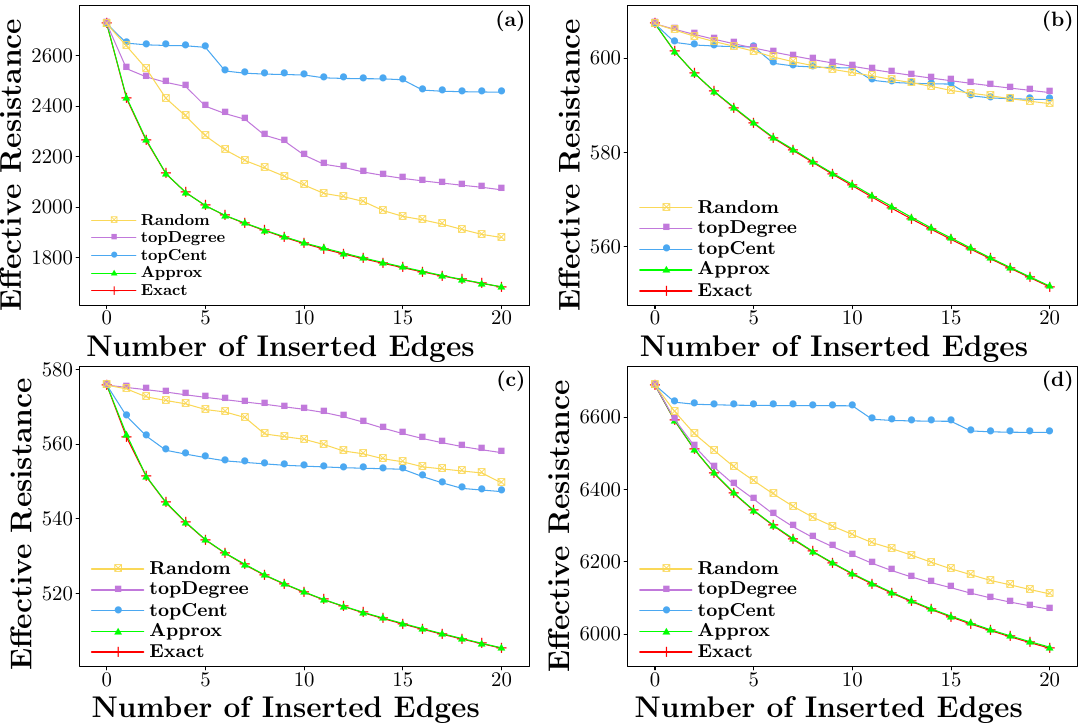}
	\caption{Effective resistance for a set $Q$ of $q=5$ leader nodes   as a function of the number $k$ of inserted edges for five heuristics,  $\ExactSM$,  $\ApproxiSM$, random, TopDegree, and TopCent on four networks: (a) Chicago,  (b) Hamster Full, (c) Facebook,  and (d) HepTh.\label{largeQ5}}
\end{figure}

\begin{figure}[tb]
	\centering
	\includegraphics[width=\linewidth]{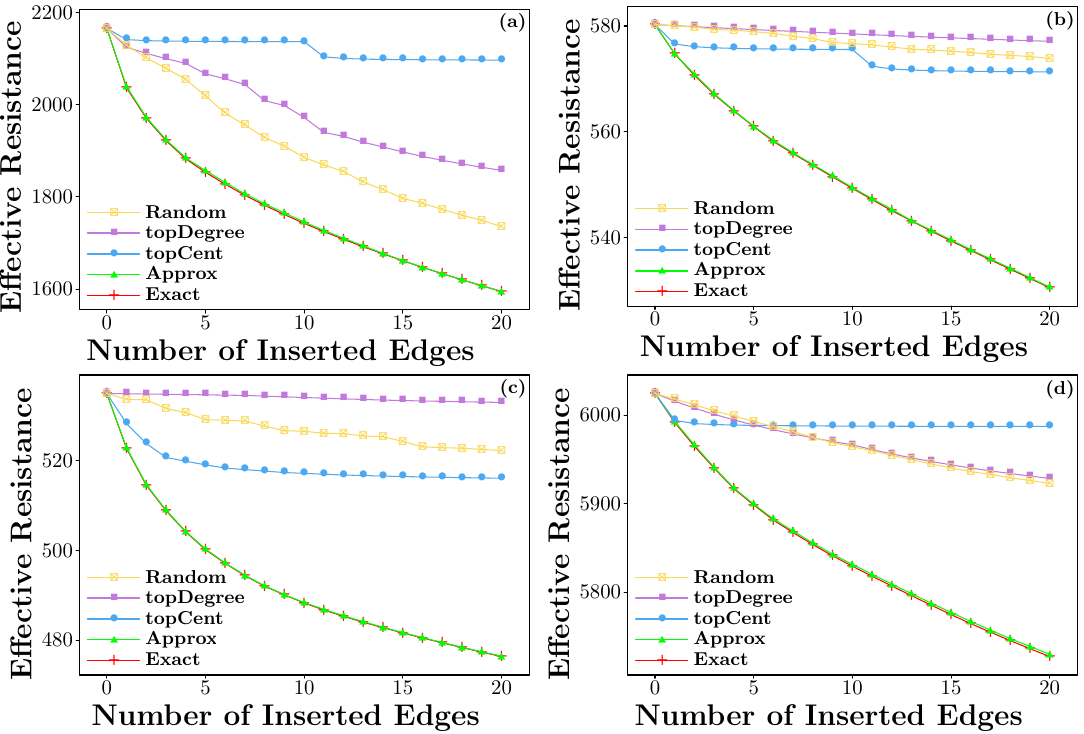}
	\caption{Effective resistance for a set $Q$ of $q=10$ leader nodes as a function of the number $k$ of inserted edges for five heuristics,  $\ExactSM$,  $\ApproxiSM$, random, TopDegree, and TopCent on four networks: (a) Chicago,  (b) Hamster Full, (c) Facebook,  and (d) HepTh.\label{largeQ10}}
\end{figure}

\subsection{Efficiency Comparison of Greedy Algorithms}

Although both  greedy algorithms  $\ApproxiSM$  and $\ExactSM$ produce good solutions, we will show  that they differ  greatly in the efficiency. For this purpose, we compare the running time of   $\ApproxiSM$  and $\ExactSM$ on some  realistic networks.  For every network, we randomly select a candidate set $Q$ of 10 target vertices, and calculate the effective resistance of $Q$ after adding $k=20$ new edges incident to vertices in $Q$ and $V \setminus Q$, and record the  running time. In Table~\ref{SetNo}, We list the running time of the two  greedy algorithms. It can be observed that for small networks with less than 18,000 nodes, $\ApproxiSM$ performs  a little slowly for most cases. However, for those networks with more than  22,000 nodes,  $\ApproxiSM$ is always much faster than $\ExactSM$, and gradually becomes faster as the node number increases.  Moreover, for those networks with more than 100,000 nodes, $\ExactSM$ fails due to the high time and memory cost,  $\ApproxiSM$ can still solve the  effective resistance. Finally,  it should be stressed that  although $\ApproxiSM$ is more efficient than $\ExactSM$, the solutions returned by both algorithms are very close to each other, as shown in Table~\ref{time}.

\begin{table}
	\centering
	\caption{Effective resistance of a group $Q$ of 10 vertices returned by algorithms  $\ApproxiSM$  and $\ExactSM$   for some real-world networks, as well as  the ratios of results of $\ApproxiSM$ to those  of $\ExactSM$.\label{time}}
	\resizebox{0.8\columnwidth}{!}{
		\begin{tabular}{ccccc}
			\toprule
			\fontsize{6.5}{8}\selectfont
			\multirow{2}*{ Network} & \multicolumn{3}{c}{Effective Resistance} \\
			\cmidrule{2-4}
			&$\ExactSM$& $\ApproxiSM$ & Ratio \\
			\midrule
			Karate &  6.2001 &  6.2909 &  1.0147 \\
			Windsurfers &  2.5033 &  2.5130 &  1.0039 \\
			Dolphins &  14.4097 &  14.5323 &  1.0085 \\
			Lesmis &  20.0719 &  20.1304 &  1.0029 \\
			Adjnoun &  23.0843 &  23.3235 &  1.0104 \\
			Celegansneural &  38.0503 &  38.1220 &  1.0019 \\
			Chicago &  1617.9612 &  1674.8639 &  1.0352 \\
			Hamster Full &  529.1304 &  529.6522 &  1.0010 \\
			Facebook &  487.8539 &  489.6902 &  1.0038 \\
			GrQc &  2992.4740 &  3034.0660 &  1.0139 \\
			Power Grid &  9518.0740 &  9836.8890 &  1.0335 \\
			High Energy &  4883.3000 &  4931.2573 &  1.0098 \\
			Reactome &  1433.9990 &  1442.4733 &  1.0059 \\
			Route Views &  4879.2456 &  4925.0080 &  1.0094 \\
			HepTh &  5660.1680 &  5694.3535 &  1.0060 \\
			Pretty Good Privacy &  14984.6100 &  15156.1290 &  1.0114 \\
			HepPh &  4452.0347 &  4460.0850 &  1.0018 \\
			AstroPh &  4515.5264 &  4523.1973 &  1.0017 \\
			Internet &  16376.9200 &  16472.5400 &  1.0058 \\
			CAIDA &  19731.6020 &  19940.6800 &  1.0106 \\
			Enron Email &  18210.8300 &  18242.4260 &  1.0017 \\
			Condensed Matter &  15298.2550 &  15329.9380 &  1.0021 \\
			Brightkite &  41192.2230 &  41240.6560 &  1.0012 \\
			\bottomrule
		\end{tabular}
	}
\end{table}

\section{Conclusions}

We examined the problem of minimizing the polarization of the leader-follower opinion dynamics in a noisy social network $G=(V,E)$ with $n$ nodes and $m$ edges, where a group $Q \subset V$ of $q$ nodes are leaders, by adding $k$ new edges incident to the nodes in $Q$. It is  a combinatorial optimization problem with an exponential computational complexity, and is equivalent to minimizing the sum of resistance distance $\mathcal{R}_Q$  between the  node group $Q$ and all other nodes.  We proved that the object function is  monotone and supermodular. We then presented two approximation algorithms for computing $\mathcal{R}_Q$: the former returns a $(1-{1}/{e})$ approximation of the optimum in time $O((n-q)^3)$, while the latter provides a $\left(1-{1}/{e}-\eps\right)$ approximation in time $\Otil (mk\eps^{-2})$. We also compared our algorithms with several potential alternative algorithms. Finally, we performed extensive experiments on real-life networks, which demonstrate our algorithms outperform the baseline methods and can often compute an approximate optimal solution. In particular, our second algorithm can yield a good approximate solution very fast, making it scalable to large-scale networks with more than one million nodes.

\section*{Acknowledgment}
Both authors are with the Shanghai Key Laboratory of Intelligent Information Processing, School of Computer Science, Fudan University, Shanghai 200433, China. This work was supported by the National Natural Science Foundation of China (Nos. U20B2051 and 61872093).
\bibliographystyle{ACM-Reference-Format}

\begin{thebibliography}{53}


\ifx \showCODEN    \undefined \def \showCODEN     #1{\unskip}     \fi
\ifx \showDOI      \undefined \def \showDOI       #1{#1}\fi
\ifx \showISBNx    \undefined \def \showISBNx     #1{\unskip}     \fi
\ifx \showISBNxiii \undefined \def \showISBNxiii  #1{\unskip}     \fi
\ifx \showISSN     \undefined \def \showISSN      #1{\unskip}     \fi
\ifx \showLCCN     \undefined \def \showLCCN      #1{\unskip}     \fi
\ifx \shownote     \undefined \def \shownote      #1{#1}          \fi
\ifx \showarticletitle \undefined \def \showarticletitle #1{#1}   \fi
\ifx \showURL      \undefined \def \showURL       {\relax}        \fi
\providecommand\bibfield[2]{#2}
\providecommand\bibinfo[2]{#2}
\providecommand\natexlab[1]{#1}
\providecommand\showeprint[2][]{arXiv:#2}

\bibitem[\protect\citeauthoryear{Abebe, Kleinberg, Parkes, and
  Tsourakakis}{Abebe et~al\mbox{.}}{2018}]%
        {AbKlPaTs18}
\bibfield{author}{\bibinfo{person}{Rediet Abebe}, \bibinfo{person}{Jon
  Kleinberg}, \bibinfo{person}{David Parkes}, {and}
  \bibinfo{person}{Charalampos~E Tsourakakis}.}
  \bibinfo{year}{2018}\natexlab{}.
\newblock \showarticletitle{Opinion dynamics with varying susceptibility to
  persuasion}. In \bibinfo{booktitle}{\emph{Proceedings of the 24th ACM SIGKDD
  International Conference on Knowledge Discovery \& Data Mining}}. ACM,
  \bibinfo{pages}{1089--1098}.
\newblock


\bibitem[\protect\citeauthoryear{Achlioptas}{Achlioptas}{2003}]%
        {Ac03}
\bibfield{author}{\bibinfo{person}{Dimitris Achlioptas}.}
  \bibinfo{year}{2003}\natexlab{}.
\newblock \showarticletitle{{Database-friendly random projections:
  Johnson-Lindenstrauss with binary coins}}.
\newblock \bibinfo{journal}{\emph{J. Comput. System Sci.}}
  \bibinfo{volume}{66}, \bibinfo{number}{4} (\bibinfo{year}{2003}),
  \bibinfo{pages}{671--687}.
\newblock


\bibitem[\protect\citeauthoryear{Altafini}{Altafini}{2013}]%
        {Al13}
\bibfield{author}{\bibinfo{person}{Claudio Altafini}.}
  \bibinfo{year}{2013}\natexlab{}.
\newblock \showarticletitle{Consensus problems on networks with antagonistic
  interactions}.
\newblock \bibinfo{journal}{\emph{IEEE Trans. Automat. Control}}
  \bibinfo{volume}{58}, \bibinfo{number}{4} (\bibinfo{year}{2013}),
  \bibinfo{pages}{935--946}.
\newblock


\bibitem[\protect\citeauthoryear{Amelkin and Singh}{Amelkin and Singh}{2019}]%
        {AmSi19}
\bibfield{author}{\bibinfo{person}{Victor Amelkin} {and}
  \bibinfo{person}{Ambuj~K Singh}.} \bibinfo{year}{2019}\natexlab{}.
\newblock \showarticletitle{Fighting opinion control in social networks via
  link recommendation}. In \bibinfo{booktitle}{\emph{Proceedings of the 25th
  ACM SIGKDD International Conference on Knowledge Discovery and Data Mining}}.
  ACM, \bibinfo{pages}{677--685}.
\newblock


\bibitem[\protect\citeauthoryear{Anderson and Ye}{Anderson and Ye}{2019}]%
        {AnYe19}
\bibfield{author}{\bibinfo{person}{Brian~DO Anderson} {and}
  \bibinfo{person}{Mengbin Ye}.} \bibinfo{year}{2019}\natexlab{}.
\newblock \showarticletitle{Recent advances in the modelling and analysis of
  opinion dynamics on influence networks}.
\newblock \bibinfo{journal}{\emph{Int. J. Autom. Comput.}}
  \bibinfo{volume}{16}, \bibinfo{number}{2} (\bibinfo{year}{2019}),
  \bibinfo{pages}{129--149}.
\newblock


\bibitem[\protect\citeauthoryear{Bamieh, Jovanovic, Mitra, and
  Patterson}{Bamieh et~al\mbox{.}}{2012}]%
        {BaJoMiPa12}
\bibfield{author}{\bibinfo{person}{Bassam Bamieh}, \bibinfo{person}{Mihailo~R
  Jovanovic}, \bibinfo{person}{Partha Mitra}, {and} \bibinfo{person}{Stacy
  Patterson}.} \bibinfo{year}{2012}\natexlab{}.
\newblock \showarticletitle{Coherence in large-scale networks:
  Dimension-dependent limitations of local feedback}.
\newblock \bibinfo{journal}{\emph{IEEE Trans. Automat. Control}}
  \bibinfo{volume}{57}, \bibinfo{number}{9} (\bibinfo{year}{2012}),
  \bibinfo{pages}{2235--2249}.
\newblock


\bibitem[\protect\citeauthoryear{Bindel, Kleinberg, and Oren}{Bindel
  et~al\mbox{.}}{2011}]%
        {BiKlOr11}
\bibfield{author}{\bibinfo{person}{David Bindel}, \bibinfo{person}{Jon
  Kleinberg}, {and} \bibinfo{person}{Sigal Oren}.}
  \bibinfo{year}{2011}\natexlab{}.
\newblock \showarticletitle{How Bad is Forming Your Own Opinion?}. In
  \bibinfo{booktitle}{\emph{Proceedings of the 2011 IEEE 52nd Annual Symposium
  on Foundations of Computer Science}}. \bibinfo{pages}{57--66}.
\newblock


\bibitem[\protect\citeauthoryear{Bu, Li, Zhang, Cao, Li, and Shi}{Bu
  et~al\mbox{.}}{2020}]%
        {BuLiZhCaLiSh20}
\bibfield{author}{\bibinfo{person}{Zhan Bu}, \bibinfo{person}{Hui-Jia Li},
  \bibinfo{person}{Chengcui Zhang}, \bibinfo{person}{Jie Cao},
  \bibinfo{person}{Aihua Li}, {and} \bibinfo{person}{Yong Shi}.}
  \bibinfo{year}{2020}\natexlab{}.
\newblock \showarticletitle{Graph {K}-means based on leader identification,
  dynamic game, and opinion dynamics}.
\newblock \bibinfo{journal}{\emph{IEEE Trans. Knowl. Data Eng.}}
  \bibinfo{volume}{32}, \bibinfo{number}{7} (\bibinfo{year}{2020}),
  \bibinfo{pages}{1348--1361}.
\newblock


\bibitem[\protect\citeauthoryear{Chen, Lijffijt, and De~Bie}{Chen
  et~al\mbox{.}}{2018}]%
        {ChLiDe18}
\bibfield{author}{\bibinfo{person}{Xi Chen}, \bibinfo{person}{Jefrey Lijffijt},
  {and} \bibinfo{person}{Tijl De~Bie}.} \bibinfo{year}{2018}\natexlab{}.
\newblock \showarticletitle{Quantifying and minimizing risk of conflict in
  social networks}. In \bibinfo{booktitle}{\emph{Proceedings of the 24th ACM
  SIGKDD International Conference on Knowledge Discovery and Data Mining}}.
  ACM, \bibinfo{pages}{1197--1205}.
\newblock


\bibitem[\protect\citeauthoryear{Clark, Alomair, Bushnell, and
  Poovendran}{Clark et~al\mbox{.}}{2014}]%
        {ClAlBuPo14}
\bibfield{author}{\bibinfo{person}{Andrew Clark}, \bibinfo{person}{Basel
  Alomair}, \bibinfo{person}{Linda Bushnell}, {and} \bibinfo{person}{Radha
  Poovendran}.} \bibinfo{year}{2014}\natexlab{}.
\newblock \showarticletitle{Minimizing convergence error in multi-agent systems
  via leader selection: A supermodular optimization approach}.
\newblock \bibinfo{journal}{\emph{IEEE Trans. Automat. Control}}
  \bibinfo{volume}{59}, \bibinfo{number}{6} (\bibinfo{year}{2014}),
  \bibinfo{pages}{1480--1494}.
\newblock


\bibitem[\protect\citeauthoryear{Clark and Poovendran}{Clark and
  Poovendran}{2011}]%
        {ClPo11}
\bibfield{author}{\bibinfo{person}{Andrew Clark} {and} \bibinfo{person}{Radha
  Poovendran}.} \bibinfo{year}{2011}\natexlab{}.
\newblock \showarticletitle{A submodular optimization framework for leader
  selection in linear multi-agent systems}. In
  \bibinfo{booktitle}{\emph{Proceedings of the 50th IEEE Conference on Decision
  and Control and European Control Conference}}. IEEE,
  \bibinfo{pages}{3614--3621}.
\newblock


\bibitem[\protect\citeauthoryear{Cohen, Kyng, Miller, Pachocki, Peng, Rao, and
  Xu}{Cohen et~al\mbox{.}}{2014}]%
        {CoKyMiPaJaPeRaXu14}
\bibfield{author}{\bibinfo{person}{Michael~B Cohen}, \bibinfo{person}{Rasmus
  Kyng}, \bibinfo{person}{Gary~L Miller}, \bibinfo{person}{Jakub~W Pachocki},
  \bibinfo{person}{Richard Peng}, \bibinfo{person}{Anup~B Rao}, {and}
  \bibinfo{person}{Shen~Chen Xu}.} \bibinfo{year}{2014}\natexlab{}.
\newblock \showarticletitle{Solving {SDD} linear systems in nearly $m
  \log^{1/2} n$ time}. In \bibinfo{booktitle}{\emph{Proceedings of the
  forty-sixth annual ACM symposium on Theory of computing}}. ACM,
  \bibinfo{pages}{343--352}.
\newblock


\bibitem[\protect\citeauthoryear{DeGroot}{DeGroot}{1974}]%
        {De74}
\bibfield{author}{\bibinfo{person}{Morris~H DeGroot}.}
  \bibinfo{year}{1974}\natexlab{}.
\newblock \showarticletitle{Reaching a consensus}.
\newblock \bibinfo{journal}{\emph{J. Amer. Statist. Assoc.}}
  \bibinfo{volume}{69}, \bibinfo{number}{345} (\bibinfo{year}{1974}),
  \bibinfo{pages}{118--121}.
\newblock


\bibitem[\protect\citeauthoryear{Dong, Ding, Mart{\'\i}nez, and Herrera}{Dong
  et~al\mbox{.}}{2017}]%
        {DoDiMa17}
\bibfield{author}{\bibinfo{person}{Yucheng Dong}, \bibinfo{person}{Zhaogang
  Ding}, \bibinfo{person}{Luis Mart{\'\i}nez}, {and} \bibinfo{person}{Francisco
  Herrera}.} \bibinfo{year}{2017}\natexlab{}.
\newblock \showarticletitle{Managing consensus based on leadership in opinion
  dynamics}.
\newblock \bibinfo{journal}{\emph{Inf. Sci.}}  \bibinfo{volume}{397}
  (\bibinfo{year}{2017}), \bibinfo{pages}{187--205}.
\newblock


\bibitem[\protect\citeauthoryear{Friedkin and Johnsen}{Friedkin and
  Johnsen}{1990}]%
        {FrJo90}
\bibfield{author}{\bibinfo{person}{Noah~E Friedkin} {and}
  \bibinfo{person}{Eugene~C Johnsen}.} \bibinfo{year}{1990}\natexlab{}.
\newblock \showarticletitle{Social influence and opinions}.
\newblock \bibinfo{journal}{\emph{J. Math. Sociol.}} \bibinfo{volume}{15},
  \bibinfo{number}{3-4} (\bibinfo{year}{1990}), \bibinfo{pages}{193--206}.
\newblock


\bibitem[\protect\citeauthoryear{Gaitonde, Kleinberg, and Tardos}{Gaitonde
  et~al\mbox{.}}{2020}]%
        {GaKlTa20}
\bibfield{author}{\bibinfo{person}{Jason Gaitonde}, \bibinfo{person}{Jon
  Kleinberg}, {and} \bibinfo{person}{Eva Tardos}.}
  \bibinfo{year}{2020}\natexlab{}.
\newblock \showarticletitle{Adversarial perturbations of opinion dynamics in
  networks}. In \bibinfo{booktitle}{\emph{Proceedings of the 21st ACM
  Conference on Economics and Computation}}. \bibinfo{pages}{471--472}.
\newblock


\bibitem[\protect\citeauthoryear{Garimella, De~Francisci~Morales, Gionis, and
  Mathioudakis}{Garimella et~al\mbox{.}}{2017}]%
        {GaDeGiMa17}
\bibfield{author}{\bibinfo{person}{Kiran Garimella}, \bibinfo{person}{Gianmarco
  De~Francisci~Morales}, \bibinfo{person}{Aristides Gionis}, {and}
  \bibinfo{person}{Michael Mathioudakis}.} \bibinfo{year}{2017}\natexlab{}.
\newblock \showarticletitle{Reducing Controversy by Connecting Opposing Views}.
  In \bibinfo{booktitle}{\emph{Proceedings of the Tenth ACM International
  Conference on Web Search and Data Mining}}. ACM, \bibinfo{pages}{81--90}.
\newblock


\bibitem[\protect\citeauthoryear{Ghosh, Boyd, and Saberi}{Ghosh
  et~al\mbox{.}}{2008}]%
        {GhBoSa08}
\bibfield{author}{\bibinfo{person}{Arpita Ghosh}, \bibinfo{person}{Stephen
  Boyd}, {and} \bibinfo{person}{Amin Saberi}.} \bibinfo{year}{2008}\natexlab{}.
\newblock \showarticletitle{Minimizing effective resistance of a graph}.
\newblock \bibinfo{journal}{\emph{SIAM Rev.}} \bibinfo{volume}{50},
  \bibinfo{number}{1} (\bibinfo{year}{2008}), \bibinfo{pages}{37--66}.
\newblock


\bibitem[\protect\citeauthoryear{Ishakian, Erd{\"o}s, Terzi, and
  Bestavros}{Ishakian et~al\mbox{.}}{2012}]%
        {IsErTeBe12}
\bibfield{author}{\bibinfo{person}{Vatche Ishakian}, \bibinfo{person}{D{\'o}ra
  Erd{\"o}s}, \bibinfo{person}{Evimaria Terzi}, {and} \bibinfo{person}{Azer
  Bestavros}.} \bibinfo{year}{2012}\natexlab{}.
\newblock \showarticletitle{A framework for the evaluation and management of
  network centrality}. In \bibinfo{booktitle}{\emph{Proceedings of the 2012
  SIAM International Conference on Data Mining}}. \bibinfo{pages}{427--438}.
\newblock


\bibitem[\protect\citeauthoryear{Izmailian, Kenna, and Wu}{Izmailian
  et~al\mbox{.}}{2013}]%
        {INK+13}
\bibfield{author}{\bibinfo{person}{N~Sh Izmailian}, \bibinfo{person}{R Kenna},
  {and} \bibinfo{person}{FY Wu}.} \bibinfo{year}{2013}\natexlab{}.
\newblock \showarticletitle{The two-point resistance of a resistor network: a
  new formulation and application to the cobweb network}.
\newblock \bibinfo{journal}{\emph{J.Phys. A: Math. Theoret.}}
  \bibinfo{volume}{47}, \bibinfo{number}{3} (\bibinfo{year}{2013}),
  \bibinfo{pages}{035003}.
\newblock


\bibitem[\protect\citeauthoryear{Johnson and Lindenstrauss}{Johnson and
  Lindenstrauss}{1984}]%
        {JoLi84}
\bibfield{author}{\bibinfo{person}{William~B Johnson} {and}
  \bibinfo{person}{Joram Lindenstrauss}.} \bibinfo{year}{1984}\natexlab{}.
\newblock \showarticletitle{{Extensions of Lipschitz mappings into a Hilbert
  space}}.
\newblock \bibinfo{journal}{\emph{Contemp. Math.}}  \bibinfo{volume}{26}
  (\bibinfo{year}{1984}), \bibinfo{pages}{189--206}.
\newblock


\bibitem[\protect\citeauthoryear{Kunegis}{Kunegis}{2013}]%
        {kunegis2013konect}
\bibfield{author}{\bibinfo{person}{J{\'e}r{\^o}me Kunegis}.}
  \bibinfo{year}{2013}\natexlab{}.
\newblock \showarticletitle{Konect: the koblenz network collection}. In
  \bibinfo{booktitle}{\emph{Proceedings of the 22nd International Conference on
  World Wide Web}}. ACM, \bibinfo{pages}{1343--1350}.
\newblock


\bibitem[\protect\citeauthoryear{Kyng and Sachdeva}{Kyng and Sachdeva}{2016}]%
        {kyng2016approximate}
\bibfield{author}{\bibinfo{person}{Rasmus Kyng} {and} \bibinfo{person}{Sushant
  Sachdeva}.} \bibinfo{year}{2016}\natexlab{}.
\newblock \showarticletitle{{Approximate Gaussian elimination for
  Laplacians-fast, sparse, and simple}}. In
  \bibinfo{booktitle}{\emph{Proceedings of the 57th Annual Symposium on
  Foundations of Computer Science}}. IEEE, \bibinfo{pages}{573--582}.
\newblock


\bibitem[\protect\citeauthoryear{Ledford}{Ledford}{2020}]%
        {Le20}
\bibfield{author}{\bibinfo{person}{Heidi Ledford}.}
  \bibinfo{year}{2020}\natexlab{}.
\newblock \showarticletitle{How Facebook, Twitter and other data troves are
  revolutionizing social science}.
\newblock \bibinfo{journal}{\emph{Nature}} \bibinfo{volume}{582},
  \bibinfo{number}{7812} (\bibinfo{year}{2020}), \bibinfo{pages}{328--330}.
\newblock


\bibitem[\protect\citeauthoryear{Leskovec and Sosi{\v{c}}}{Leskovec and
  Sosi{\v{c}}}{2016}]%
        {LeSo16}
\bibfield{author}{\bibinfo{person}{Jure Leskovec} {and} \bibinfo{person}{Rok
  Sosi{\v{c}}}.} \bibinfo{year}{2016}\natexlab{}.
\newblock \showarticletitle{{SNAP: A} general-purpose network analysis and
  graph-mining library}.
\newblock \bibinfo{journal}{\emph{ACM Trans. Intell. Syst. Technol.}}
  \bibinfo{volume}{8}, \bibinfo{number}{1} (\bibinfo{year}{2016}),
  \bibinfo{pages}{1}.
\newblock


\bibitem[\protect\citeauthoryear{Li, Patterson, Yi, and Zhang}{Li
  et~al\mbox{.}}{2020}]%
        {LiPaYiZh20}
\bibfield{author}{\bibinfo{person}{Huan Li}, \bibinfo{person}{Stacy Patterson},
  \bibinfo{person}{Yuhao Yi}, {and} \bibinfo{person}{Zhongzhi Zhang}.}
  \bibinfo{year}{2020}\natexlab{}.
\newblock \showarticletitle{Maximizing the number of spanning trees in a
  connected graph}.
\newblock \bibinfo{journal}{\emph{IEEE Trans. Inf. Theory}}
  \bibinfo{volume}{66}, \bibinfo{number}{2} (\bibinfo{year}{2020}),
  \bibinfo{pages}{1248--1260}.
\newblock


\bibitem[\protect\citeauthoryear{Li and Schild}{Li and Schild}{2018}]%
        {LiSc18}
\bibfield{author}{\bibinfo{person}{Huan Li} {and} \bibinfo{person}{Aaron
  Schild}.} \bibinfo{year}{2018}\natexlab{}.
\newblock \showarticletitle{Spectral Subspace Sparsification}. In
  \bibinfo{booktitle}{\emph{Proceedings of 2018 IEEE 59th Annual Symposium on
  Foundations of Computer Science}}. IEEE, \bibinfo{pages}{385--396}.
\newblock


\bibitem[\protect\citeauthoryear{Liu, Xu, Lu, Chen, and Zeng}{Liu
  et~al\mbox{.}}{2021}]%
        {LiXuLuChZe21}
\bibfield{author}{\bibinfo{person}{Hui Liu}, \bibinfo{person}{Xuanhong Xu},
  \bibinfo{person}{Jun-An Lu}, \bibinfo{person}{Guanrong Chen}, {and}
  \bibinfo{person}{Zhigang Zeng}.} \bibinfo{year}{2021}\natexlab{}.
\newblock \showarticletitle{Optimizing pinning control of complex dynamical
  networks based on spectral properties of grounded {L}aplacian matrices}.
\newblock \bibinfo{journal}{\emph{IEEE Trans. Syst., Man, Cybern., Syst.}}
  \bibinfo{volume}{51}, \bibinfo{number}{2} (\bibinfo{year}{2021}),
  \bibinfo{pages}{786--796}.
\newblock


\bibitem[\protect\citeauthoryear{Luca, Fabio, Paolo, and Asuman}{Luca
  et~al\mbox{.}}{2014}]%
        {VaFaFr14}
\bibfield{author}{\bibinfo{person}{Vassio Luca}, \bibinfo{person}{Fagnani
  Fabio}, \bibinfo{person}{Frasca Paolo}, {and} \bibinfo{person}{Ozdaglar
  Asuman}.} \bibinfo{year}{2014}\natexlab{}.
\newblock \showarticletitle{Message Passing Optimization of Harmonic Influence
  Centrality}.
\newblock \bibinfo{journal}{\emph{IEEE Trans. Control Netw. Syst.}}
  \bibinfo{volume}{1}, \bibinfo{number}{1} (\bibinfo{year}{2014}),
  \bibinfo{pages}{109--120}.
\newblock


\bibitem[\protect\citeauthoryear{Ma, Zheng, and Wang}{Ma et~al\mbox{.}}{2016}]%
        {MaZhWa16}
\bibfield{author}{\bibinfo{person}{Jingying Ma}, \bibinfo{person}{Yuanshi
  Zheng}, {and} \bibinfo{person}{Long Wang}.} \bibinfo{year}{2016}\natexlab{}.
\newblock \showarticletitle{Topology selection for multi-agent systems with
  opposite leaders}.
\newblock \bibinfo{journal}{\emph{Syst. \& Control Lett.}}
  \bibinfo{volume}{93}, \bibinfo{number}{7} (\bibinfo{year}{2016}),
  \bibinfo{pages}{43--49}.
\newblock


\bibitem[\protect\citeauthoryear{Mackin and Patterson}{Mackin and
  Patterson}{2019}]%
        {MaPa19}
\bibfield{author}{\bibinfo{person}{Erika Mackin} {and} \bibinfo{person}{Stacy
  Patterson}.} \bibinfo{year}{2019}\natexlab{}.
\newblock \showarticletitle{Maximizing diversity of opinion in social
  networks}. In \bibinfo{booktitle}{\emph{Proceedings of 2019 American Control
  Conference}}. IEEE, \bibinfo{pages}{2728--2734}.
\newblock


\bibitem[\protect\citeauthoryear{Matakos, Terzi, and Tsaparas}{Matakos
  et~al\mbox{.}}{2017}]%
        {MaTeTs17}
\bibfield{author}{\bibinfo{person}{Antonis Matakos}, \bibinfo{person}{Evimaria
  Terzi}, {and} \bibinfo{person}{Panayiotis Tsaparas}.}
  \bibinfo{year}{2017}\natexlab{}.
\newblock \showarticletitle{Measuring and Moderating Opinion Polarization in
  Social Networks}.
\newblock \bibinfo{journal}{\emph{Data. Min. Knowl. Disc.}}
  \bibinfo{volume}{31}, \bibinfo{number}{5} (\bibinfo{year}{2017}),
  \bibinfo{pages}{1480--1505}.
\newblock


\bibitem[\protect\citeauthoryear{Medya, Silva, Singh, Basu, and Swami}{Medya
  et~al\mbox{.}}{2018}]%
        {MeSiSiBaSw18}
\bibfield{author}{\bibinfo{person}{Sourav Medya}, \bibinfo{person}{Arlei
  Silva}, \bibinfo{person}{Ambuj Singh}, \bibinfo{person}{Prithwish Basu},
  {and} \bibinfo{person}{Ananthram Swami}.} \bibinfo{year}{2018}\natexlab{}.
\newblock \showarticletitle{Group centrality maximization via network design}.
  In \bibinfo{booktitle}{\emph{Proceedings of the 2018 SIAM International
  Conference on Data Mining}}. SIAM, \bibinfo{pages}{126--134}.
\newblock


\bibitem[\protect\citeauthoryear{Meyer}{Meyer}{1973}]%
        {Me73}
\bibfield{author}{\bibinfo{person}{Carl~D Meyer, Jr}.}
  \bibinfo{year}{1973}\natexlab{}.
\newblock \showarticletitle{Generalized inversion of modified matrices}.
\newblock \bibinfo{journal}{\emph{SIAM J. Appl. Math.}} \bibinfo{volume}{24},
  \bibinfo{number}{3} (\bibinfo{year}{1973}), \bibinfo{pages}{315--323}.
\newblock


\bibitem[\protect\citeauthoryear{Musco, Musco, and Tsourakakis}{Musco
  et~al\mbox{.}}{2018}]%
        {MuMuTs18}
\bibfield{author}{\bibinfo{person}{Cameron Musco}, \bibinfo{person}{Christopher
  Musco}, {and} \bibinfo{person}{Charalampos~E Tsourakakis}.}
  \bibinfo{year}{2018}\natexlab{}.
\newblock \showarticletitle{Minimizing polarization and disagreement in social
  networks}. In \bibinfo{booktitle}{\emph{Proceedings of the 2018 World Wide
  Web Conference}}. \bibinfo{pages}{369--378}.
\newblock


\bibitem[\protect\citeauthoryear{Nemhauser, Wolsey, and Fisher}{Nemhauser
  et~al\mbox{.}}{1978}]%
        {NeWoFi78}
\bibfield{author}{\bibinfo{person}{George~L Nemhauser},
  \bibinfo{person}{Laurence~A Wolsey}, {and} \bibinfo{person}{Marshall~L
  Fisher}.} \bibinfo{year}{1978}\natexlab{}.
\newblock \showarticletitle{An analysis of approximations for maximizing
  submodular set functions}.
\newblock \bibinfo{journal}{\emph{Math. Program.}} \bibinfo{volume}{14},
  \bibinfo{number}{1} (\bibinfo{year}{1978}), \bibinfo{pages}{265--294}.
\newblock


\bibitem[\protect\citeauthoryear{Noorazar}{Noorazar}{2020}]%
        {No20}
\bibfield{author}{\bibinfo{person}{Hossein Noorazar}.}
  \bibinfo{year}{2020}\natexlab{}.
\newblock \showarticletitle{Recent advances in opinion propagation dynamics: a
  2020 survey}.
\newblock \bibinfo{journal}{\emph{Eur. Phys. J. Plus}} \bibinfo{volume}{135},
  \bibinfo{number}{6} (\bibinfo{year}{2020}), \bibinfo{pages}{521}.
\newblock


\bibitem[\protect\citeauthoryear{Ohara, Saito, Kimura, and Motoda}{Ohara
  et~al\mbox{.}}{2017}]%
        {OhSaKiMo17}
\bibfield{author}{\bibinfo{person}{Kouzou Ohara}, \bibinfo{person}{Kazumi
  Saito}, \bibinfo{person}{Masahiro Kimura}, {and} \bibinfo{person}{Hiroshi
  Motoda}.} \bibinfo{year}{2017}\natexlab{}.
\newblock \showarticletitle{Maximizing network performance based on group
  centrality by creating most effective $k$-links}. In
  \bibinfo{booktitle}{\emph{2017 IEEE International Conference on Data Science
  and Advanced Analytics}}. IEEE, \bibinfo{pages}{561--570}.
\newblock


\bibitem[\protect\citeauthoryear{Parotsidis, Pitoura, and Tsaparas}{Parotsidis
  et~al\mbox{.}}{2016}]%
        {PaPiTs16}
\bibfield{author}{\bibinfo{person}{Nikos Parotsidis},
  \bibinfo{person}{Evaggelia Pitoura}, {and} \bibinfo{person}{Panayiotis
  Tsaparas}.} \bibinfo{year}{2016}\natexlab{}.
\newblock \showarticletitle{Centrality-aware link recommendations}. In
  \bibinfo{booktitle}{\emph{Proceedings of the 9th ACM International Conference
  on Web Search and Data Mining}}. ACM, \bibinfo{pages}{503--512}.
\newblock


\bibitem[\protect\citeauthoryear{Patterson and Bamieh}{Patterson and
  Bamieh}{2010}]%
        {PaBa10}
\bibfield{author}{\bibinfo{person}{Stacy Patterson} {and}
  \bibinfo{person}{Bassam Bamieh}.} \bibinfo{year}{2010}\natexlab{}.
\newblock \showarticletitle{Leader selection for optimal network coherence}. In
  \bibinfo{booktitle}{\emph{Proceedings of the 49th IEEE Conference on Decision
  and Control}}. IEEE, \bibinfo{pages}{2692--2697}.
\newblock


\bibitem[\protect\citeauthoryear{Perra and Rocha}{Perra and Rocha}{2019}]%
        {PeRo19}
\bibfield{author}{\bibinfo{person}{Nicola Perra} {and} \bibinfo{person}{Luis~EC
  Rocha}.} \bibinfo{year}{2019}\natexlab{}.
\newblock \showarticletitle{Modelling opinion dynamics in the age of
  algorithmic personalisation}.
\newblock \bibinfo{journal}{\emph{Sci. Rep.}} \bibinfo{volume}{9},
  \bibinfo{number}{1} (\bibinfo{year}{2019}), \bibinfo{pages}{1--11}.
\newblock


\bibitem[\protect\citeauthoryear{Smith and Christakis}{Smith and
  Christakis}{2008}]%
        {SmCh08}
\bibfield{author}{\bibinfo{person}{Kirsten~P. Smith} {and}
  \bibinfo{person}{Nicholas~A. Christakis}.} \bibinfo{year}{2008}\natexlab{}.
\newblock \showarticletitle{Social networks and health}.
\newblock \bibinfo{journal}{\emph{Annu. Rev. Sociol.}} \bibinfo{volume}{34},
  \bibinfo{number}{1} (\bibinfo{year}{2008}), \bibinfo{pages}{405--429}.
\newblock


\bibitem[\protect\citeauthoryear{Spielman and Teng}{Spielman and Teng}{2014}]%
        {SpTe14}
\bibfield{author}{\bibinfo{person}{D. Spielman} {and} \bibinfo{person}{S.
  Teng}.} \bibinfo{year}{2014}\natexlab{}.
\newblock \showarticletitle{Nearly Linear Time Algorithms for Preconditioning
  and Solving Symmetric, Diagonally Dominant Linear Systems}.
\newblock \bibinfo{journal}{\emph{SIAM J. Matrix Anal. Appl.}}
  \bibinfo{volume}{35}, \bibinfo{number}{3} (\bibinfo{year}{2014}),
  \bibinfo{pages}{835--885}.
\newblock


\bibitem[\protect\citeauthoryear{Spielman and Srivastava}{Spielman and
  Srivastava}{2011}]%
        {SpSr11}
\bibfield{author}{\bibinfo{person}{Daniel~A Spielman} {and}
  \bibinfo{person}{Nikhil Srivastava}.} \bibinfo{year}{2011}\natexlab{}.
\newblock \showarticletitle{Graph sparsification by effective resistances}.
\newblock \bibinfo{journal}{\emph{SIAM J. Comput.}} \bibinfo{volume}{40},
  \bibinfo{number}{6} (\bibinfo{year}{2011}), \bibinfo{pages}{1913--1926}.
\newblock


\bibitem[\protect\citeauthoryear{Taylor}{Taylor}{1968}]%
        {Ta68}
\bibfield{author}{\bibinfo{person}{Michael Taylor}.}
  \bibinfo{year}{1968}\natexlab{}.
\newblock \showarticletitle{Towards a mathematical theory of influence and
  attitude change}.
\newblock \bibinfo{journal}{\emph{Hum. Relat.}} \bibinfo{volume}{21},
  \bibinfo{number}{2} (\bibinfo{year}{1968}), \bibinfo{pages}{121--139}.
\newblock


\bibitem[\protect\citeauthoryear{Xiao, Boyd, and Kim}{Xiao
  et~al\mbox{.}}{2007}]%
        {XiBoKi07}
\bibfield{author}{\bibinfo{person}{Lin Xiao}, \bibinfo{person}{Stephen Boyd},
  {and} \bibinfo{person}{Seung-Jean Kim}.} \bibinfo{year}{2007}\natexlab{}.
\newblock \showarticletitle{Distributed average consensus with
  least-mean-square deviation}.
\newblock \bibinfo{journal}{\emph{J. Parallel. Distrib. Comput.}}
  \bibinfo{volume}{67}, \bibinfo{number}{1} (\bibinfo{year}{2007}),
  \bibinfo{pages}{33--46}.
\newblock


\bibitem[\protect\citeauthoryear{Xu, Bao, and Zhang}{Xu et~al\mbox{.}}{2021}]%
        {XuBaZh21}
\bibfield{author}{\bibinfo{person}{Wanyue Xu}, \bibinfo{person}{Qi Bao}, {and}
  \bibinfo{person}{Zhongzhi Zhang}.} \bibinfo{year}{2021}\natexlab{}.
\newblock \showarticletitle{Fast evaluation for relevant quantities of opinion
  dynamics}. In \bibinfo{booktitle}{\emph{Proceedings of The Web Conference}}.
  ACM, \bibinfo{pages}{2037--2045}.
\newblock


\bibitem[\protect\citeauthoryear{Xu, Zhu, Guan, Zhang, and Zhang}{Xu
  et~al\mbox{.}}{2022}]%
        {XuZhGuZhZh20}
\bibfield{author}{\bibinfo{person}{Wanyue Xu}, \bibinfo{person}{Liwang Zhu},
  \bibinfo{person}{Jiale Guan}, \bibinfo{person}{Zuobai Zhang}, {and}
  \bibinfo{person}{Zhongzhi Zhang}.} \bibinfo{year}{2022}\natexlab{}.
\newblock \showarticletitle{Effects of Stubbornness on Opinion Dynamics}. In
  \bibinfo{booktitle}{\emph{Proceedings of the 31st ACM International
  Conference on Information \& Knowledge Management}}.
  \bibinfo{pages}{2321--2330}.
\newblock


\bibitem[\protect\citeauthoryear{Yi, Castiglia, and Patterson}{Yi
  et~al\mbox{.}}{2021}]%
        {YiCaPa21}
\bibfield{author}{\bibinfo{person}{Yuhao Yi}, \bibinfo{person}{Timothy
  Castiglia}, {and} \bibinfo{person}{Stacy Patterson}.}
  \bibinfo{year}{2021}\natexlab{}.
\newblock \showarticletitle{Shifting opinions in a social network through
  leader selection}.
\newblock \bibinfo{journal}{\emph{IEEE Transactions on Control of Network
  Systems}} \bibinfo{volume}{8}, \bibinfo{number}{3} (\bibinfo{year}{2021}),
  \bibinfo{pages}{1116--1127}.
\newblock


\bibitem[\protect\citeauthoryear{Zhang, Zhang, and Chen}{Zhang
  et~al\mbox{.}}{2021}]%
        {ZhZhCh21}
\bibfield{author}{\bibinfo{person}{Zuobai Zhang}, \bibinfo{person}{Zhongzhi
  Zhang}, {and} \bibinfo{person}{Guanrong Chen}.}
  \bibinfo{year}{2021}\natexlab{}.
\newblock \showarticletitle{Minimizing spectral radius of non-backtracking
  matrix by edge removal}. In \bibinfo{booktitle}{\emph{Proceedings of the 30th
  ACM International Conference on Information \& Knowledge Management}}.
  \bibinfo{publisher}{ACM}, \bibinfo{pages}{2657--2667}.
\newblock


\bibitem[\protect\citeauthoryear{Zhou and Zhang}{Zhou and Zhang}{2021}]%
        {ZhZh21}
\bibfield{author}{\bibinfo{person}{Xiaotian Zhou} {and}
  \bibinfo{person}{Zhongzhi Zhang}.} \bibinfo{year}{2021}\natexlab{}.
\newblock \showarticletitle{Maximizing Influence of Leaders in Social
  Networks}. In \bibinfo{booktitle}{\emph{Proceedings of the 27th ACM SIGKDD
  Conference on Knowledge Discovery \& Data Mining}}.
  \bibinfo{pages}{2400--2408}.
\newblock


\bibitem[\protect\citeauthoryear{Zhou, Zhu, Li, and Zhang}{Zhou
  et~al\mbox{.}}{2023}]%
        {ZhZhLiZh23}
\bibfield{author}{\bibinfo{person}{Xiaotian Zhou}, \bibinfo{person}{Liwang
  Zhu}, \bibinfo{person}{wei Li}, {and} \bibinfo{person}{Zhongzhi Zhang}.}
  \bibinfo{year}{2023}\natexlab{}.
\newblock \showarticletitle{A Sublinear time algorithm for opinion optimization
  in directed social networks via edge recommendation}. In
  \bibinfo{booktitle}{\emph{Proceedings of the 29th ACM SIGKDD International
  Conference on Knowledge Discovery \& Data Mining}}. ACM,
  \bibinfo{pages}{3593--3602}.
\newblock


\bibitem[\protect\citeauthoryear{Zhu, Bao, and Zhang}{Zhu
  et~al\mbox{.}}{2021}]%
        {ZhBaoZh21}
\bibfield{author}{\bibinfo{person}{Liwang Zhu}, \bibinfo{person}{Qi Bao}, {and}
  \bibinfo{person}{Zhongzhi Zhang}.} \bibinfo{year}{2021}\natexlab{}.
\newblock \showarticletitle{Minimizing Polarization and Disagreement in Social
  Networks via Link Recommendation}. In \bibinfo{booktitle}{\emph{Proceedings
  of the 35th Conference on Advances in Neural Information Processing
  Systems}}. \bibinfo{pages}{2072--2084}.
\newblock


\end{thebibliography}
\balance

\end{document}